\documentclass[a4paper,11pt]{article}

\usepackage{blindtext,titlefoot}

\usepackage{amsfonts,amsthm,amssymb,amsmath}
\usepackage[title,titletoc,toc]{appendix}
\usepackage[utf8]{inputenc}
\usepackage[affil-it]{authblk}
\usepackage{xcolor}
\usepackage{tikz}
\newcommand{\RN}[1]{%
  \textup{\uppercase\expandafter{\romannumeral#1}}%
}
\newcommand\COMP{\hbox{C\kern -.58em {\raise .54ex \hbox{$\scriptscriptstyle |$}}
\kern-.55em {\raise .53ex \hbox{$\scriptscriptstyle |$}} }}
\newcommand\NN{\hbox{I\kern-.2em\hbox{N}}}
\newcommand\RR{\hbox{I\kern-.2em\hbox{R}}}
\newcommand\sRR{{\it \hbox{I\kern-.2em\hbox{R}}}}
\newcommand\QQ{\hbox{I\kern-.53em\hbox{Q}}}
\newcommand\PP{\hbox{I\kern-.53em\hbox{P}}}
\newcommand\EE{\hbox{I\kern-.53em\hbox{E}}}
\newcommand\ZZ{{{\rm Z}\kern-.28em{\rm Z}}}
\newcommand\be{\begin{equation}}
\newcommand\ee{\end{equation}}
\usepackage{scalerel,stackengine}
\stackMath
\newcommand\reallywidehat[1]{%
\savestack{\tmpbox}{\stretchto{%
  \scaleto{%
    \scalerel*[\widthof{\ensuremath{#1}}]{\kern.1pt\mathchar"0362\kern.1pt}%
    {\rule{0ex}{\textheight}}%WIDTH-LIMITED CIRCUMFLEX
  }{\textheight}% 
}{2.4ex}}%
\stackon[-6.9pt]{#1}{\tmpbox}%
}

\newtheorem{theorem}{Theorem}[section]

\newtheorem{lemma}[theorem]{Lemma}

\newtheorem{definition}[theorem]{Definition}

\makeatletter
\newcommand*\bigcdot{\mathpalette\bigcdot@{.5}}
\newcommand*\bigcdot@[2]{\mathbin{\vcenter{\hbox{\scalebox{#2}{$\m@th#1\bullet$}}}}}
\makeatother
\newcommand{\is}{\bigcdot }

\newcommand{\E}{\mathbb E}

\def \id{1\!\!1}

\def \Lbrack {[\![}
\def \Rbrack {]\!]}

\numberwithin{equation}{section}

\DeclareMathOperator*{\loc}{loc}

\begin{document}

%%%\title{Risk decomposition for multiple deaths and/or defaults frameworks}

\title{Log-optimal portfolio after a random time: Existence, description and sensitivity analysis}

\author{ Ferdoos Alharbi\\
Mathematical and Statistical Sciences Dept.\\
University of Alberta, Edmonton, Canada\\
and\\
Saudi Electronic University, Saudi Arabia\\
\and
  Tahir Choulli\\ 
  Mathematical and Statistical Sciences Dept.\\
University of Alberta, Edmonton, Canada}

\maketitle\unmarkedfntext{
This research is supported by NSERC (through grant NSERC RGPIN04987)\\
Address correspondence to Tahir Choulli, Department of Mathematical and Statistical Sciences, University of Alberta, 632 Central Academic Building,
Edmonton, Canada; e-mail: tchoulli@ualberta.ca}

\begin{abstract} In this paper, we consider an {\it informational} market model with two flows of informations. The smallest flow $\mathbb{F}$, which is available to all agents, is the filtration of the initial market model $(S, \mathbb{F},P)$, where $S$ is the assets' prices and $P$ is a probability measure. The largest flow $\mathbb{G}$ contains additional information about the occurrence of a random time $\tau$. This setting covers credit risk theory where $\tau$ models the default time of a firm, and life insurance where $\tau$ represents the death time of an insured. For the model $(S-S^{\tau},\mathbb{G},P)$, we address the log-optimal portfolio problem in many aspects. In particular, we answer the following questions and beyond: 1) What are the necessary and sufficient conditions for the existence of log-optimal portfolio of the model under consideration? 2) what are the various type of risks induced by $\tau$ that affect this portfolio and how? 3) What are the factors that completely describe the sensitivity of the log-portfolio to the parameters of $\tau$? The answers to these questions and other related discussions definitely complement the work of Choulli and Yansori \cite{ChoulliYansori2} which deals with the stopped model $(S^{\tau},\mathbb{G})$. 
\end{abstract}


\section{Introduction} In this paper, we consider an initial market model  represented by the triplet $(S,\mathbb F,P)$, where $S$ represents the discounted prices for $d$-stocks, $\mathbb F$ is the ``public" information that is available to most agents, and $P$ is a probability measure. To this initial model, we add a random time $\tau$ that might not be seen through $\mathbb F$ when it occurs (mathematically speaking $\tau$ might not be an $\mathbb F$-stopping time). In this context, we adopt the progressive enlargement of filtration to model the larger information that includes both $\mathbb F$ and $\tau$. For the obtained new informational system that lives after $\tau$, denoted by $(S-S^{\tau},\mathbb G,P)$, our ultimate goal lies in measuring the impact of $\tau$ on log-optimal portfolio,  no matter what is the pair $(S,\mathbb F)$ and no matter how it is related to $\tau$, which is an honest time with some mild assumption. In order to be more precise in our objective, we  recall the definition of log-optimal and num\'eraire portfolios, the two portfolios that are intimately related to the logarithm utility.  To this end, we denote by $W^{\theta}$ the wealth process of the portfolio $\theta$.
%%%%%%%%%%%%%%%%%%%%%%%%%%%%%%%%%%%%%%%%%%%%%%%%%%%%%%%%%%%%
\begin{definition}\label{NP/LogOP}  Let $(X, \mathbb H, Q)$ be a market model, where $X$ is the assets' price process, $\mathbb H$ is a filtration, and $Q$ is a probability measure. Consider a fixed investment horizon $T\in(0,+\infty)$, and a portfolio $\theta^*$.\\
{\rm{(a)}} $\theta^*$ is a {\it num\'eraire portfolio} for $(X, \mathbb H, Q)$ if  $W^{\theta^*}>0$ and
\begin{eqnarray}\label{NP}
{{W^{\theta}}\over{W^{\theta^*}}}\ \mbox{is an $(\mathbb H, Q)$-supermartingale, for any portfolio $\theta$ with $W^{\theta}\geq 0$}.\hskip 0.5cm
\end{eqnarray}
{\rm{(b)}} $\theta^*$ is called a {\it log-optimal portfolio} for $(X, \mathbb H, Q)$ if $\theta^*\in \Theta(X,\mathbb H, Q)$ and 
\begin{eqnarray}
 u_T(X,\mathbb H, Q):=\sup_{\theta\in\Theta}E_Q\left[\ln(W^{\theta}_T)\right]= E_Q\left[\ln(W^{\theta^*}_T)\right],\label{LogInfinite}\end{eqnarray}
where $E_Q[.]$ is the expectation under $Q$, and $\Theta:=\Theta(X,\mathbb H, Q)$ is given by 
\begin{eqnarray}\label{AdmissibleSet0}
\hskip -0.6cm \Theta(X,\mathbb H, Q):=\Bigl\{\mbox{portfolio}\ \theta\ :\ W^{\theta}> 0\quad \mbox{and}\quad E_Q\left[\vert \ln(W^{\theta}_T)\vert \right]<+\infty\Bigr\}.\end{eqnarray}
\end{definition}

The problem of maximization of expected log utility from terminal wealth, defined in (\ref{LogInfinite})-(\ref{AdmissibleSet0}), received a lot of attention in the literature, even though it is a particular case of the utility maximization theory problem. This latter problem is addressed at various levels of generality, see  Cvitanic et al. \cite{CSW}, Karatzas and Wang \cite{Karatzas},  Karatzas  and Zitkovic \cite{KZ}, Kramkov and Schachermayer \cite{KW99}, Merton \cite{merton71,merton73}, and the references therein to cite a few. The num\'eraire portfolio was introduced --up to our knowledge-- in Long \cite{Long}, where $W^{\theta}/W^{\theta^*}$ is required to be a martingale, while  Definition \ref{NP/LogOP}-(a) goes back to Becherer \cite[Definition 4.1]{Becherer}. Then these works were extended and investigated extensively in different directions in Choulli et al. \cite{ChoulliDengMa}, Christensen and Larsen \cite{ChristensenLarsen2007}, G\"oll and Kallsen \cite{GollKallsen},  Hulley and Schweizer \cite{HulleySchweizer}, Kardaras and Karatzas \cite{KardarasKaratzas}, and the references therein. For more precise relationship between the num\'eraire and log-optimal portfolios, we refer the reader to \cite{ChoulliYansori1} and the references therein.\\

 Thus, our setting falls into the topic of the {\it portfolio problem under asymmetries of information}. The literature about this topic can be divided into two major cases. The first main case, which is the most studied in the literature, is known in the finance and mathematical finance literature as {\it the insider trading} setting. For this insider framework, log-optimal portfolios are extensively studied and we refer the reader to Amendinger et al. \cite{amendingerimkellerschweizer98}, Ankirchner et al. \cite{ADImkeller}, Ankirchner and Imkeller \cite{AImkeller}, Corcuera et al. \cite{JImkellerKN}, Grorud and Pontier \cite{GrorudPontier}, Pikovsky and Karatzas \cite{pikovskykaratzas96}, Kohatsu-Higa and Sulem \cite{kohatsusulem06}, and the references therein to cite a few. Most of this literature focuses on two intimately related questions on the log-optimal portfolio for $(S,{\mathbb G}^*,P)$, where ${\mathbb G}^*$ is the initial enlargement of $\mathbb F$ with a random variable $L$ that represents the extra knowledge. In fact, under some assumption on the pair $(L, \mathbb F)$, frequently called Jacod's assumption, the existence of the log-optimal portfolio and the evaluation of  the {\bf  increment of expected logarithm-utility from terminal wealth} (denoted hereafter by IEU$(S,{\mathbb G}^*, \mathbb F)$) for $(S,{\mathbb G}^*,P)$ and $(S,\mathbb F,P)$ represent the core contribution of these papers, where it is proven that  
\begin{equation}\label{InsiderFormula}
\mbox{IEU}(S,{\mathbb G}^*, \mathbb F):=u_T(S, {\mathbb G}^*,P)-u_T(S, \mathbb F,P)=\mbox{relative entropy}(P\big| Q^*).
\end{equation}
Hence, in this insider setting, the log-optimal portfolio for $(S,{\mathbb G}^*,P)$  exists if and only if $P$ has a finite entropy with respect to $Q^*$, which is an explicitly described probability measure associated to $L$. In particular, the quantity  $\mbox{IEU}(S,{\mathbb G}^*, \mathbb F)$ is always a true gain due to the advantage of knowing fully $L$ by the investor endowed with the flow $\mathbb G^*$. The formula (\ref{InsiderFormula}) was initially derived by Pikovsky and Karatzas \cite{pikovskykaratzas96} for the Brownian filtration, and Amendinger et al. \cite{amendingerimkellerschweizer98} extended it to models driven by general continuous local martingales, where the authors connect this formula with Shannon entropy of $L$ for some models. The Shannon concept was further studied by Ankirchner et al. \cite{ADImkeller} afterwards, where the authors show its important role in measuring the impact of inside-information on log-optimal portfolios. Other related and interesting works, we cite Corcuera et al. \cite{JImkellerKN} and Kahatsu-Higa and Yamazato \cite{Kohatsu2011}, and Ankirchner et al. \cite{ADImkeller}. In this latter paper, the authors consider arbitrary filtrations  $\mathbb F$ and $\mathbb G^*$ such that $\mathbb F\subset \mathbb G^*$, while assuming the continuity of $S$  and the {\it existence of a drift information condition} on the pair $(\mathbb F,\mathbb G^*)$. \\

The second major case, in contrast to the insider setting which uses the initial enlargement of $\mathbb{F}$, suggests to add the extra information over time as it occurs, and this leads to say that $\mathbb{G}$ is the progressive enlargement of filtration with $\tau$. This is our current framework, in which we complement \cite{ChoulliYansori1} that deals with the sub-model $(S^{\tau},\mathbb{G},P)$. On the one hand,  a log-optimal portfolio is a num\'eraire portfolio, see Choulli and Yansori \cite{ChoulliYansori1,ChoulliYansori2}. On the other hand, thanks to Choulli et al. \cite{ChoulliDengMa} and Kardaras and Karatzas \cite{KardarasKaratzas} that connect the existence of a num\'eraire portfolio to the concept of No-Unbounded-Profit-with-bounded-risk (NUPBR hereafter), and the recent works of Aksamit et al. \cite{aksamitetal18,ACDJ3} and Choulli and Deng \cite{CD1} on NUPBR for the model $(S-S^{\tau},\mathbb G,P)$, the problem of the existence of the num\'eraire portfolio for $(S-S^{\tau},\mathbb G,P)$ is completely understood. For the log-optimal portfolio of $(S-S^{\tau}, \mathbb{G},P)$, the situation is more challenging, and the problem of its existence is the first major obstacle. To address this, we appeal to the explicit description of the set of deflators of $(S-S^{\tau}, \mathbb{G},P)$, recently developed by Choulli and Yansori \cite{ChoulliAlharbi}, and answer the following question. 
\begin{equation}\label{Q2}
\mbox{For which $(S, \tau)$, does the log-optimal portfolio of $(S^{\tau}, \mathbb G,P)$ exist?}\end{equation}
It is worth mentioning that this existence question is much deeper and general than the corresponding one addressed in the insider setting. Indeed, in our framework, there is no hope for (\ref{InsiderFormula}) to hold in its current form, and only a practical answer to (\ref{Q2}) will allow us to answer the question below.
%%%%%%%%%%%%%%%%%%%%%%%%%%%%%%%%%%%%%%%%%%%%%%%%%%%%%%%%%%%%%%%%%%%%%%%%%%%%%
\begin{equation}\label{Q3}
\begin{cases}
\mbox{What are the {\it informational conditions} on $\tau$ }\\
\mbox{for the existence of the log-optimal portfolio of $(S-S^{\tau}, \mathbb G,P)$ }\\ 
 \mbox{when it already exists for $(S,\mathbb F,P)$?}\end{cases}
\end{equation}

%%%%%%%%%%%%%%%%%%%%%%%%%%%%%%%%%%%%%%%%%%%%%%%%%%%%%%%%%%%%%%%%% the main question
For our  setting, the {\it increment of expected logarithmic-utility} between $(S-S^{\tau}, \mathbb G)$ and $(S,\mathbb F)$, denoted by IE$\mbox{U}_{\mbox{after}}(S, \tau, \mathbb F)$, is defined by 
\begin{equation}\label{Delta(S, Tau)}
\mbox{IEU}_{\mbox{after}}(S,\tau, \mathbb F):=\Delta_T(S, \tau, \mathbb F):=u_T(S-S^{\tau}, \mathbb G,P)-u_T(S, \mathbb F,P),\end{equation}
and is affected by many factors. Hence, we will address the question of 
\begin{equation}\label{Q5}
\mbox{what are the factors that explain the sensitivity of IEU$_{\mbox{after}}(S,\tau,\mathbb F)$ to $\tau$}?\end{equation}
Our definition for the utility increment, given in (\ref{Delta(S, Tau)}), stems from our main goal of measuring the impact of $\tau$ on the base model $(S,\mathbb F,P)$ via log utility. To answer (\ref{Q5}), we address the explicit computation of the log-optimal portfolio using $\mathbb F$-observable processes, and answer the following question.
\begin{equation}\label{Q4}
\mbox{How can the log-optimal portfolio of $(S-S^{\tau},\mathbb G,P)$ be described via $\mathbb F$?}\end{equation}
%%%%%%%%%%%%%%%%%%%%%%%%%%%%%%%%%%%%%%%%%%%%%%%%%%%%%%%%%We give necessary and sufficient conditions in terms of processes 
This paper contains four sections including the current one. Section  \ref{Section2} presents the mathematical and the financial model besides the corresponding required notation and some preliminaries that are important herein. Section \ref{Section3} focuses on the existence of the log-optimal portfolio and the duality. Section \ref{Section4} describes explicitly both the log-optimal portfolio using the $\mathbb F$-predictable characteristics of the model, and discusses its financial applications and consequences. The paper contains an appendix where some proofs are relegated and  some technical (new and existing) results are detailed. 
%%%%%%%%%%%%%%%%%%%%%%%%%%%%%%%%%%%%%%%%%%%%%%%%%%%%%%%%%%%%%%%%%%%%%%%%%%%%%%%%%%%
\section{The mathematical framework and preliminaries}\label{Section2}
Throughout the paper, we suppose given a complete probability space $(\Omega, {\cal{F}},P)$. Then any filtration $\mathbb{H}$ on this space will be supposed to satisfy the usual conditions (i.e. $\mathbb{H}$ is  right-continuous and complete).
\subsection{General notation}For any any filtration  $\mathbb{H}$ on this space, we denote ${\cal A}(\mathbb H)$ (respectively ${\cal M}(\mathbb H)$) the set
of $\mathbb H$-adapted processes with $\mathbb H$-integrable variation (respectively that are $\mathbb H$-uniformly integrable martingale).
For any process $X$, we denote by $^{o,\mathbb H}X$  (respectively $^{p,\mathbb H}X$)  the
$\mathbb H$-optional (respectively $\mathbb H$-predictable) projection of $X$. For an increasing process $V$, we denote $V^{o,\mathbb H}$ (respectively $V^{p,\mathbb H}$) its dual $\mathbb H$-optional (respectively $\mathbb H$-predictable) projection. For a filtration $\mathbb H$, ${\cal O}(\mathbb H)$, ${\cal P}(\mathbb H)$ and  $\mbox{Prog}(\mathbb H)$ represent the $\mathbb H$-optional, the $\mathbb H$-predictable and the $\mathbb H$-progressive $\sigma$-fields  respectively on $\Omega\times[0,+\infty[$. For an $\mathbb H$-semimartingale $X$, we denote by $L(X,\mathbb H)$ the set of all $X$-integrable processes in the Ito's sense, and for $H\in L(X,\mathbb H)$, the resulting integral is one dimensional $\mathbb H$-semimartingale denoted by $H\is X:=\int_0^{\cdot} H_udX_u$. If ${\cal C}(\mathbb H)$ 
is a set of processes that are adapted to $\mathbb H$,
then ${\cal C}_{\loc}(\mathbb H)$ --except when it is stated otherwise-- is the set of processes, $X$,
for which there exists a sequence of $\mathbb H$-stopping times,
$(T_n)_{n\geq 1}$, that increases to infinity and $X^{T_n}$ belongs to ${\cal C}(\mathbb H)$, for each $n\geq 1$. For any $\mathbb H$-semimartingale, $L$, the Doleans-Dade stochastic exponential denoted by ${\cal E}(L)$, is the unique solution to the SDE: $dX = X_{-} dL,$ $ X_0= 1,$ given by 
\begin{equation}\label{S-exponential}
{\cal E}_t (L) = \exp \big ( L_t- L_0 - {1 \over 2} {\langle L^c \rangle}_{t} \big ) \prod_{0 <  s \leq t} \big( 1 + \Delta L_s \big) e^{-\Delta L_s}.
\end{equation}
\subsection{The mathematical model}Our mathematical and financial model has two principal components. The first component, called throughout the rest of the paper as the initial market model is specified by the pair $(S,\mathbb{F})$. Here $\mathbb{F}$ is filtration representing the ``public" flow of information, which is available to all agents, and $S$ is a d-dimensional $\mathbb{F}$-semimartingale which models the price processes of $d$ risky assets.  The second main component is the random time $\tau$, which might represent the default time of firm in credit risk theory, or the death time of an insured in life insurance, or the occurrence time of an event that might impact the initial market model. The occurrence of $\tau$, in general, might not be observable through $\mathbb{F}$, or mathematically this translates to $\tau$ might not be an $\mathbb{F}$-stopping time, see the literature about credit risk or that of life insurance.  Hence, in order to model the flow of information for the whole model $(S,\mathbb{F}, \tau)$, we consider the following
\begin{equation}\label{GGtildem}
G_t := {^{o,\mathbb F}(I_{\Lbrack0,\tau\Lbrack})_t}=P(\tau > t | {\cal F}_t),\quad \widetilde{G}_t := {^{o,\mathbb F}(I_{\Lbrack0,\tau\Rbrack})}=P(\tau \ge t | {\cal F}_t),
\quad \mbox{ and } \quad \ m := G + D^{o,\mathbb F}.
\end{equation}
 The processes $G$ and $\widetilde G$ are known as Az\'ema supermartingale, while $m$ is a BMO $\mathbb F$-martingale. For more details about these, we refer the reader to  \cite[paragraph 74, Chapitre XX]{dellacheriemeyer92}. The pair $(\widetilde{G}, G)$ or equivalently the pair $(G, D^{o,\mathbb{F}})$ are the parametrization  of $\tau$ through the flow $\mathbb{F}$. For detailed discussion about  why we consider this progressive enlargement of $\mathbb{F}$ with $\tau$ instead, we refer the reader to \cite{ChoulliYansori2} and the references therein. \\
 
 Our resulting informational model, that we will investigate throughout the paper,  is $(S, \mathbb{G}, P)$. For this model, our
  main goal resides in studying the log-optimal portfolio deeply, and mainly answer the questions singled out in the introduction. As Choulli and Yansori \cite{ChoulliYansori2} already addressed fully the sub-model $(S^{\tau}, \mathbb{G}, P)$, and due to the myopic feature of the log utility, the main contribution of this paper lies in focusing on the sub-model $(S-S^{\tau}, \mathbb{G}, P)$. It is known in the probabilistic literature, that for general $\tau$, the process $S-S^{\tau}$ might not even be a $\mathbb{G}$-semimartingale, and hence in this case there is no chance for the log-optimal portfolio to exist. To avoid this technical problem and difficulty, which is not in the scope of this paper, we suppose throughout the paper that $\tau$ is an honest random time. This honest time concept is mathematically defined by the following.
\begin{definition}\label{honesttime}
 A random time $\tau$ is called an ${\mathbb F}$-honest time if, for any $t$, there exists an ${\mathbb{F}}_t$-measurable random variable $\tau_t$ such that $\tau \id_{\{\tau<t\}} = \tau_t \id_{\{\tau<t\}}.$
\end{definition}
%%%%%%%%%%%%%%%%%%%%%%%%%%%%%%%%%%%%%%%%%%%%%%%%%%%%%%%%%%%%%%%%%%%%%%%%%%%%%%%%%
\subsection{Some useful preliminaries}
This subsection recalls some useful results and/or some definitions from the literature. 
 \begin{theorem} \label{OptionalDecompoTheorem}  Suppose $\tau$ is an honest time. Then the following assertions hold.\\
 {\rm{(a)}} For any $\mathbb F$-local martingale $M$, the process
 \begin{eqnarray}\label{honestMhat}
{\cal T}^{(a)}(M):=I_{\Rbrack\tau,+\infty\Lbrack}\is{M}+{{I_{\Rbrack\tau,+\infty\Lbrack}}\over{1- G}}\is [m,M]+{{I_{\Rbrack\tau,+\infty\Lbrack}}\over{1-G_{-}}}\is \left(\sum \Delta M(1-G_{-}) I_{\{\widetilde G=1>G_{-}\}}\right)^{p,\mathbb F}
 \end{eqnarray}
 is a $\mathbb G$-local martingale.\\
 {\rm{(b)}} For any $M \in {\mathcal M}_{loc} (\mathbb F)$, the process 
\begin{equation}\label{Glocalmartingaleaftertau}
  \widetilde{M}^{(a)}:= I_{\Rbrack \tau, +\infty \Lbrack} \is M + {1 \over 1-G_{-}} I_{\Rbrack \tau, +\infty \Lbrack} \is \langle m, M\rangle^{\mathbb{F}}\quad\mbox{is a $\mathbb G$-local martingale.}
\end{equation}
 %\\
% {\rm{(b)}} If $\displaystyle\Lbrack\tau\Rbrack\subset\bigcup_{n=1}^{+\infty}\Lbrack T_n\Rbrack$, where $(T_n)_n$ is a sequence of  $\mathbb F$-stopping times with disjoint graphs, then
% \begin{eqnarray}\label{thinMhat}
%  {\cal T}_a(M):=M-{{\id_{\Rbrack0,\tau\Rbrack}}\over{\widetilde G}}\is [m,M]+\id_{\Rbrack0,\tau\Rbrack}\is \left(\sum \Delta M \id_{\{\widetilde G=0<G_{-}\}}\right)^{p,\mathbb F}
% -\sum_{n\geq 1} \id_{C_n}{{\id_{\Rbrack T_n\,+\infty\Lbrack}}\over{z^{(n)}}}\is [z^{(n)}, M]
% \end{eqnarray}
 % is a  is a $\mathbb G$-local martingale.
 \end{theorem}
 
 %%%%%%%%%%%%%%%%%%%%%%%%%%%%%%%%%%%%%%%%%%%%%%%%%%%%%%%%%%%%%%%%%%%% Explicit description of all deflators   %%%%%%%%%%%%% %%%%%%%%%%%%%%%%%%%%%%%%%%%%%%%%%%%%%%%%%%%%%%%%%%%%%%%%%%%%%
We recall the mathematical definition of deflators, as these are the ``dual" processes to the wealth processes $W^{\theta}$, and hence play central role in solving the dual problem. 
 \begin{definition}\label{DeflatorDefinition} Consider the model $(X, \mathbb H, Q)$, where $\mathbb H$ is a filtration, $Q$ is a probability, and $X$ is a $(Q,\mathbb H)$-semimartingale. Let $Z$ be a process.\\
  We call $Z$ a local martingale deflator for $(X,Q,\mathbb H)$ if $Z>0$ and there exists a real-valued and $\mathbb H$-predictable process $\varphi$ such that $0<\varphi\leq 1$ and both $Z$  and $Z(\varphi\is X)$ are $\mathbb H$-local martingales under $Q$.  Throughout the paper, the set of these local martingale deflators will be denoted by ${\cal Z}_{loc}(X,Q,\mathbb H)$. The set of all deflators will be denoted by ${\cal D}(X,Q,\mathbb H)$. When $Q=P$, for the sake of simplicity, we simply omit the probability in notations and terminology.
\end{definition}
Thanks to \cite[Theorem 2.1]{ChoulliYansori1}, the log-optimal portfolio for a model $(X,\mathbb{H},P)$ is intimately related to the subset of ${\cal D}(X,\mathbb H)$ given by 
 \begin{equation}\label{logdeflatorset}
     {\cal D}_{log} (X, \mathbb H) := \left \{ Z \in {\cal D}( X, \mathbb H ) \quad \vert \quad  \sup_{t \geq 0 } E [-\ln (Z_t) ] < +\infty \right \}.
 \end{equation}
We end this subsection by borrowing an important result, from \cite{ChoulliAlharbi}, on explicit description of all deflators for the model $(S-S^{\tau},\mathbb G)$ in terms of deflators for the initial model $(S, \mathbb F)$. To this end, throughout the paper, we assume the following assumptions
\begin{equation}\label{Assumptions4Tau}
\tau\quad\mbox{is a finite honest time such that}\quad G_{\tau}<+\infty\quad P\mbox{-a.s. and}\quad \left\{\widetilde{G}=1>G_{-}\right\}=\emptyset.
\end{equation}
and we consider the following processes
\begin{equation}\label{m(a)}
m^{(1)} := -(1-G_{-})^{-1} I_{\{G_{-}<1\}}\is m\quad\mbox{and}\quad {S}^{(1)}:=I_{\{G_{-}<1\}}\is S.
\end{equation}

 \begin{theorem}\label{GeneralDefaltorDescription4afterTau}
 Suppose that assumptions (\ref{Assumptions4Tau}) hold, and let  $Z^{\mathbb G}$ be a process such that $(Z^{\mathbb G})^{\tau}\equiv 1$. Then the following assertions are equivalent.\\
{\rm{(a)}} $Z^{\mathbb G}$ is a deflator for $(S-S^{\tau}, \mathbb G)$  (i.e., $Z^{\mathbb G}\in {\cal D}(S-S^{\tau}, \mathbb G)$).\\
{\rm{(b)}} There exists a unique pair $\left(K^{\mathbb F}, V^{\mathbb F}\right)$ such that $K^{\mathbb F}\in {\cal M}_{loc}(\mathbb F)$,  $V^{\mathbb F}$ is an $\mathbb F$-predictable RCLL and nondecreasing process such that
\begin{eqnarray*}V^{\mathbb F}_0=K^{\mathbb F}_0=0,\quad {\cal E}(K^{\mathbb F}){\cal E}(-V^{\mathbb F})\in {\cal D}({S}^{(1)}, \mathbb F),\end{eqnarray*}
\begin{equation}\label{repKG1a}
Z^{\mathbb G}={\cal E}(K^{\mathbb G}){\cal E}(-I_{\Rbrack\tau,+\infty\Lbrack}\is{V}^{\mathbb F}),\ K^{\mathbb G}=I_{\Rbrack\tau,+\infty\Lbrack}\is {\cal T}^{(a)}(K^{\mathbb F})+(1-G_{-})^{-1}I_{\Rbrack\tau,+\infty\Lbrack}\is {\cal T}^{(a)}(m).\end{equation}
{\rm{(c)}} There exists a unique $Z^{\mathbb F}\in{\cal D}({S}^{(1)}, \mathbb F)$ such that  \begin{equation}\label{repKGMultiGEneral}
Z^{\mathbb G}={{Z^{\mathbb F}/(Z^{\mathbb F})^{\tau}}\over{{\cal E}(-I_{\Rbrack\tau,+\infty\Lbrack}(1-G_{-})^{-1}\is m)}}.\end{equation}
 \end{theorem}
 We end this section with the following definition of portfolio rate. Recall that a portfolio for the model $(X,\mathbb{H},Q)$ is any $\mathbb{H}$-predictable $\theta$ ($\mathbb{H}\in\{\mathbb{F},\mathbb{G}\}$), such that $\theta$ is $X$-integrable in the semimartingale sense, and the corresponding wealth process with initial value one is given by $W^{\theta}=1+\theta\is{X}$. 
 \begin{definition} Let $\theta$ be a portfolio for $(X,\mathbb{H},Q)$ such that both $W^{\theta}$ and $W^{\theta}_{-}$ are positive (i.e. $W^{\theta}>0$ and $W^{\theta}_{-}>0$ ). Then we associate to the portfolio $\theta$, the portfolio rate $\varphi$ given by
 $$\varphi:=\theta/W^{\theta}_{-}.$$
  It exists, it is $X$-integrable and $\varphi\in {\cal{L}}(X,\mathbb{H},Q)$, where
 \begin{equation}\label{L(X,H)}
  {\cal{L}}(X,\mathbb{H},Q):=\Bigl\{\varphi\ \mathbb{H}\mbox{-predictable}:\ \varphi\Delta{X}>-1\Bigr\}.
 \end{equation}
 When $Q=P$, we simply omit the probability in notation and write $ {\cal{L}}(X,\mathbb{H})$. 
 \end{definition}

%%%%%%%%%%%%%%%%%%%%%%%%%%%%%%%%%%%%%%%%%%%%%%%%%%%%%%%%%%%%%%%%re, 
%%%%%%%%%%%%%%%%%%%%%%%%%%%%%%%%%%%%%%%%%%%%%%%%%%%%%%%%%%%%%
%%%%%%%%%%%%%%%%%%%%%%%%%%%%%%%%%%%%%%%%%%%%%%%%%%%%%%%%%%%%%%%%%%%%%%%%%%%%%%%%%%%
%%%%%%%%%%%%%%%%%%%%%%%%%%%%%%%%%%%%%%%%%%%%%%%%%%%%%%%%%%%%%%%%%%%%%%%%%%%%%%%
%%%%%%%%%%%%%%%%%%%%%%%%%%%%%%%%%%%%%%%%%%%%%%%%%%%%%%%%%%%%%%%%%%%%%%%%
%%%%%%%%%%%%%%%%%%%%%%%%%%%%%%%%%%%%%%%%%%%%%%%%%%%%%%%%%%%%%%%%%%%%%%
 %%%%%%%%%%%%%%%%%%%%%%%%%%%%%%%%%%%%%%%%%%%%%%%%%%%%%%%%%%%%%%%%%%%%%%
 %%%%%%%%LOG_OPTIMAL PORTFOLIO PROBLEM %%%%%%%%%%%%%%%%%%%%%%%%%%%%%%%%%%%%%%%%%
 %%%%%%%%%%%%%%%%%%%%%%%%%%%%%%%%%%%%%%%%%%%%%%%%%%%%%%%%%%%%%%%%%%%%%%%
 
 \section{Log-optimal portfolio for $(S-S^{\tau},\mathbb{G})$: Existence and duality}\label{Section3}
 In this section, we address the existence of the log-optimal portfolio after random time. Or equivalently, in virtue of \cite[Theorem 2.1]{ChoulliYansori1}, the existence of the solution to the dual minimization problem 
 \begin{equation}\label{dualaftertau}
 \min_{Z^{\mathbb G} \in {\cal D}_{log}(S-S^{\tau}, \mathbb G)} \E \left [ -\ln (Z^{\mathbb G}_T)\right].
 \end{equation}
The solution to this dual problem, when it exists, it will {\it naturally} involve the information theoretical concept(s) such as Hellinger or entropy that we recall below.
  \begin{definition}\label{hellinger}
 Consider a filtration $\mathbb H$, and let $N$ be an $\mathbb H$-local martingale such that $1+ \Delta N > 0 $, and denoted by $N^c$ its continuous local martingale part. \\
 1) We call a Hellinger process of order zero for $N$, denoted by $h^{(0)}(N, \mathbb H)$, the process $h^{(0)}(N,\mathbb H):= \left(H^{(0)}(N, \mathbb H) \right)^{p, \mathbb H}$ when this projection exists, where 
 \begin{align}
 H^{(0)} (N, \mathbb H) := {1 \over 2} \langle N^{c}\rangle^{\mathbb H} + \sum (\Delta N - \ln(1+\Delta N)). 
 \end{align}
2) We call an entropy-Hellinger process for $N$, denoted by $h^{(E)}(N, \mathbb{H})$, the process $h^{(E)}(N, \mathbb H):= (H^{(E)}(N,\mathbb H))^{p, \mathbb H}$ when this projection exists, where
 \begin{align}
     H^{(E)} (N, \mathbb H) := {1 \over 2} \langle N^{c}\rangle^{\mathbb H} + \sum ((1+\Delta N) \ln(1+\Delta N) - \Delta N). 
 \end{align}
 3) Let $Q^1$ and $Q^2$ be two probabilities such that $Q^1\ll Q^2$. If $Q^i_T:=Q^i\big|_{{\cal H}_T}$ denotes the restriction of $Q^i$ to ${\cal H}_T$  ($i=1,2$), then  
\begin{eqnarray}\label{entropy}
{\cal H}_{\mathbb H}(Q^1_T\big| Q^2_T):=E_{Q^2}\left[{{dQ^1_T}\over{dQ^2_T}}\ln\left({{dQ^1_T}\over{dQ^2_T}}\right)\right].\end{eqnarray}
 \end{definition}
The following theorem, which is our main contribution of this section,  characterizes completely and in various manners the existence of log-optimal portfolio for $(S-S^{\tau},\mathbb{G})$. 
%%%%%%%%%%%%%%%%%%%%%%%%%%%%%%%%%%%%%%%%%%%%%%%%%%%%%%%%%%%%%%%%%%%%%%%%%%%%%%%%%
 \begin{theorem}\label{existencetheorem}
 Suppose (\ref{Assumptions4Tau})  holds, and consider $m^{(1)}$ and ${S}^{(1)}$ defined in (\ref{m(a)}). Then the following assertions are equivalent.\\ {\rm{(a)}} Log-optimal portfolio for $(S-S^{\tau}, \mathbb G)$ exists, and $\widetilde{\varphi}^{\mathbb G}$ denotes its portfolio rate. \\
 {\rm{(b)}} There exists  $K^{\mathbb F} \in {\mathcal M}_{loc} (\mathbb F)$ and a RCLL, nondecreasing, and $\mathbb F$-predictable process $V$ such 
 \begin{equation}\label{Condition4Existence0}
 K^{\mathbb F}_0 = V_0 = 0,\quad Z^{\mathbb F} := {\cal E}(K^{\mathbb F})\exp (-V) \in {\cal D}({S}^{(1)}, \mathbb F)\end{equation}
 and
 \begin{equation}\label{existencecond1}
    E \left [ (1-\widetilde{G}) \is \left(H^{(0)} (K^{\mathbb F},\mathbb F)+V+h^{(E)}( m^{(1)},\mathbb F)+\langle K^{\mathbb F}, m^{(1)} \rangle^{\mathbb F}\right)_T\right]<+\infty.
 \end{equation}
 {\rm{(c)}} There exits a unique solution to (\ref{dualaftertau}). I.e. there exists unique $\widetilde{Z}^{\mathbb G} \in {\cal D}(S-S^{\tau}, \mathbb G)$ such that
 \begin{equation}
    \inf_{Z^{\mathbb G} \in {\cal D}(S-S^{\tau}, \mathbb G) } E \left[-\ln (Z^{\mathbb G}_T)\right ] =E \left [ -\ln (\widetilde{Z}_{T}^{\mathbb G})\right] < +\infty.
 \end{equation}
 {\rm(d)} There exists a unique $\widetilde{Z}^{\mathbb F} \in {\cal D}({S}^{(1)}, \mathbb F)$ such that 
 \begin{equation}
       \inf_{Z^{\mathbb G} \in {\cal D}(S-S^{\tau}, \mathbb G) } E \left[-\ln (Z^{\mathbb G}_T)\right ] = E \left [-\ln \dfrac{\widetilde{Z}_T^{\mathbb F}/\widetilde{Z}_{T \wedge \tau}^{\mathbb F}}{{\cal E}(I_{\Rbrack \tau, \infty \Lbrack} \is{m}^{(1)})}\right ].
 \end{equation}
 Furthermore, when the triplet $(\widetilde{\varphi}^{\mathbb G},\widetilde{Z}^{\mathbb G}, \widetilde{Z}^{\mathbb F})$ exists, then it satisfies the following, 
 \begin{equation}
     {\cal E}(\widetilde{\varphi}^{\mathbb G} \is (S-S^{\tau})) = \dfrac{1}{\widetilde{Z}^{\mathbb G}} =  \dfrac{(\widetilde{Z}^{\mathbb F})^{\tau} {\cal E}(I_{\Rbrack \tau , \infty \Lbrack} \is{m}^{(1)})} {\widetilde{Z}^{\mathbb F}}.
 \end{equation}
 \end{theorem}
 %%%%%%%%%%%%%%%%%%%%%%%%%%%%%%%%%%%%%%%%%%%%%%%%%%%%%%%%%%%%%%%%%%%%%%%%
 %%%%%%%%%%%%%%%%%%%%%%%%%%%%%%%%%%%%%%%%%%%%%%%%%%%%%%%%%%%%%%%%%%%%%
%%%%%%%%%%%%%%%%%%%%%%%%%%%%%%%%%%%%%%%%%%%%%% {\bf DISCUSSIONS ABOUT the theorem.....}\\
 The proof of Theorem \ref{existencetheorem} relies on the following lemma, which is interesting in itself. 
 
 \begin{lemma}\label{deflator4hellinger}
Suppose that (\ref{Assumptions4Tau}) holds, and let $m^{(1)}$ and ${S}^{(1)}$ given in (\ref{m(a)}).
%%\begin{equation}\label{m(a)}m^{(a)} := -(1-G_{-})^{-1} I_{\{G_{-}<1\}}\is m,\quad \overline{S}:=I_{\{G_{-}<1\}}\is S.\end{equation}
 Then the following hold \\
{\rm{(a)}} Both processes $(1-G_{-})^{-2}I_{\Rbrack \tau,\infty \Lbrack} \is \langle m \rangle^{\mathbb F}$ and $H^{(0)} (m^{(1)},\mathbb{F})$ are ${\mathbb F}$-locally integrable, and 
\begin{eqnarray}
 \left ( {I_{\Rbrack \tau, \infty \Lbrack} \over (1-G_{-})^2} \is \langle m \rangle^{\mathbb F} -I_{\Rbrack \tau, +\infty \Lbrack} \is H^{(0)} (m^{(1)}, {\mathbb F}) \right )^{p, \mathbb F}  =(1-G_{-}) \is h^{(E)}(m^{(1)},\mathbb F).\label{m(1)2Hellinger}
\end{eqnarray}
{\rm{(b)}} For any $Z^{\mathbb{G}}\in {\cal{D}}_{log}(S-S^{\tau},\mathbb{G})$, there exists $Z^{\mathbb{F}}\in {\cal{D}}(S^{(1)},\mathbb{F})$ such that 
\begin{equation*}
E\left[-\ln\left(Z^{\mathbb{G}}_T\right)\right]\geq E\left[{{Z^{\mathbb{F}}_T/Z^{\mathbb{F}}_{T\wedge\tau}}\over{{\cal E}_{T} \left (- I_{\Rbrack \tau, \infty \Lbrack}(1-G_{-})^{-1}  \is m \right)}}\right].
\end{equation*}
{\rm{(c)}}The following equality holds
\begin{equation}\label{optimalZF}
    \inf_{Z^{\mathbb G} \in {\cal D}(S-S^{\tau}, \mathbb G)} E \left [ - \ln (Z^{\mathbb G}_T) \right] = \inf_{Z^{\mathbb F} \in {\cal D}({S}^{(1)}, \mathbb F)} E \left[ -\ln \left (\frac{Z_T^{\mathbb F}/ Z_{T \wedge \tau}^{\mathbb F}}{{\cal E}_{T} \left (- I_{\Rbrack \tau, \infty \Lbrack}(1-G_{-})^{-1}  \is m \right)} \right ) \right] 
\end{equation}
{\rm{(d)}}
If  $Z^{\mathbb F} \in {\cal D}({S}^{(1)}, \mathbb F)$ such that $Z^{\mathbb F}/\left((Z^{\mathbb F} )^{\tau}{\cal E}(-I_{\Rbrack \tau, \infty \Lbrack} \is{m}^{(1)})\right)\in {\cal D}_{log}(S-S^{\tau}, \mathbb G)$, then there exists a nondecreasing and $\mathbb F$-predictable process $V$ such 
$$K^{\mathbb F}_0 = V_0 = 0,\quad Z^{\mathbb F} := {\cal E}(K^{\mathbb F})\exp (-V),$$
 and 
\begin{eqnarray}
&& E \left[ -\ln \left (\frac{Z_T^{\mathbb F}/Z_{T\wedge \tau}^{\mathbb F}}{{\cal E}_{T} \left (I_{\Rbrack \tau, \infty \Lbrack} \is {m}^{(1)} \right)}\right )\right] \nonumber \\
&&= E \left [(1-\widetilde{G})\is \left(V+h^{(E)}(m^{(1)},\mathbb F)+ H^{(0)} (K^{\mathbb F},\mathbb F)+\langle K^{\mathbb F},  m^{(1)} \rangle^{\mathbb F}  \right)_T\right].\label{fromZF2KV}
\end{eqnarray}
\end{lemma}
The proof of this lemma is relegated to Appendix \ref{proof4lemmas}, while below we give the proof of Theorem  \ref{existencetheorem}.
%%%%%%%%%%%%%%%%%%%%%%%%%%%%%%%%%%%%%%%%%%%%%%%%%%%%%%%%%%%%%%%%%%%
%%%%%%%%%%%%%%%%%%%%%%%%%%%%%%%%%%%%%%%%%%%%%%%%%%%%%%%%%%%%%%%%%%%%%% 
  \begin{proof}[Proof of Theorem \ref{existencetheorem}] The proof of (a)$\iff$(b) follows immediately from \cite[ Theorem 3.2]{ChoulliYansori1} applied to the model $(X,\mathbb{H})=(S-S^{\tau},\mathbb{G})$. If assertion (c) holds, then by combining (\ref{optimalZF}) with Lemma \ref{deflator4hellinger}-(b) applied to $\widetilde{Z}^{\mathbb{G}}$, assertion (d) follows immediately. Hence, this proves (c) $\Longrightarrow$ (d). It is clear that if assertion (d) holds, then the existence of the solution  to (\ref{dualaftertau}) exists and takes the form of $\widetilde{Z}^{\mathbb{G}}=\dfrac{\widetilde{Z}^{\mathbb F}/(\widetilde{Z}^{\mathbb F})^{\tau}}{{\cal E}(I_{\Rbrack \tau, \infty \Lbrack} \is{m}^{(1)})}$, while the uniqueness of the solution is due to the strict concavity of  the logarithm function. This proves (d) $\Longrightarrow$ (c). Thus the rest of this proof focuses on proving (c) $\Longleftrightarrow$ (b). To this end, we recall that in virtue of  \cite[ Theorem 3.2]{ChoulliYansori1}, the problem (\ref{dualaftertau})  admits a solution if and only if  ${\cal D}_{log}(S-S^{\tau}, \mathbb G)\not=\emptyset$, or equivalently $\inf_{Z^{\mathbb G} \in {\cal D}(S-S^{\tau})} E \left[  -\ln (Z_T^{\mathbb G})\right ] < +\infty.$ Thus, when assertion (b) holds, the equality (\ref{fromZF2KV}) in Lemma \ref{deflator4hellinger} allows us to conclude that assertion (c) holds, and the proof of (b) $\Longrightarrow$ (c) is complete. To prove the reverse, we assume that assertion (c) holds and consider  ${Z}^{\mathbb G} \in {\cal D}_{log} (S-S^{\tau}, \mathbb G)\subset {\cal D} (S-S^{\tau}, \mathbb G)$. Thus, Theorem \ref{GeneralDefaltorDescription4afterTau} guarantees the existence of $Z^{\mathbb F} \in {\cal D} (S^{(1)}, \mathbb{F})$ such that
  $$
  {Z}^{\mathbb G}={{Z^{\mathbb F}/(Z^{\mathbb F})^{\tau}}\over{{\cal E}\left(-I_{\Rbrack \tau, \infty \Lbrack} \is{m}^{(1)}\right)}}\in  {\cal D}_{log} (S-S^{\tau}, \mathbb G).$$
  Therefore, by applying Lemma \ref{deflator4hellinger}-(d) to $Z^{\mathbb F} $, assertion (b) follows immediately. This proves  (c) $\Longrightarrow$ (b), and ends the proof of the theorem.
 \end{proof}
 %%%%%%%%%%%%%%%%%%%%%%%%%%%%%%%%%%%%%%%%%%%%%%%%%%%%%%%%%%%%%%%%%%%%%
 %%%%%%%%%%%%%%%%%%%%%%%%%%%%%%%%%%%%%%%%%%%%%%%%%%%%%%%%%%%%%%%%%%%%%%%%
 In the rest of this section, we discuss a direct consequence of Theorem \ref{existencetheorem}, which is important in itself. It is in fact an application of Theorem \ref{existencetheorem}, which gives  a sufficient condition in terms of information theoretical concept, for $(S-S^{\tau}, \mathbb G)$ to admit the log-optimal portfolio when $({S}^{(1)},\mathbb{F})$ does already.
 \begin{theorem}\label{sufficientcond}
  Suppose that (\ref{Assumptions4Tau})  holds, and $({S}^{(1)}, \mathbb F)$ admits the log-optimal portfolio. Then the log-optimal portfolio for $(S-S^{\tau}, \mathbb G)$ exists if 
  \begin{equation}\label{existencecond2} 
      E \Bigl[  \int_0^T (1- G_{s-})dh_s^{(E)}(m^{(1)}, \mathbb{F})\Bigr] < +\infty. 
  \end{equation}
 \end{theorem}
 
 \begin{proof}
 On the one hand, thanks to a combination of \cite[Theorem 2.1]{ChoulliYansori1} and  \cite[Proposition B.2]{ChoulliYansori2}, we deduce that $({S}^{(1)}, \mathbb F)$ admits the log-optimal portfolio if and only if there exist $K^{\mathbb F} \in {\cal M}_{loc}(\mathbb F)$ and a RCLL, nondecreasing, and $\mathbb F$-predictable process $V$ such that 
 %$K^{\mathbb F}_0= V_0 = 0$, and $Z^{\mathbb F} := {\cal E}(K^{\mathbb F})\exp (-V) \in {\cal D}({S}^{(1)}, \mathbb F)$, and 
 \begin{equation*}
  K^{\mathbb F}_0= V_0 = 0,\ Z^{\mathbb F} := {\cal E}(K^{\mathbb F})e^{-V} \in {\cal D}({S}^{(1)}, \mathbb F)\ \mbox{and}\   E \left [ -\ln (Z_T^{\mathbb{F}}) \right ] = E \left [ V_T + H^{(0)}_T  (K^{\mathbb F}, \mathbb F)\right ] < + \infty .
 \end{equation*}
 On the other hand, thanks to \cite [Lemma B.1]{ChoulliYansori2} and the condition $E[H^{(0)}_T  (K^{\mathbb F}, \mathbb F)]<+\infty$, we deduce that  $\sup_{0 \leq t \leq T} \vert K_t^{\mathbb F} \vert \in L^{1}(P)$, (or equivalently $E[[K^{\mathbb F} ,K^{\mathbb F}]^{1/2}_T]<+\infty$). A combination of this condition with $m$ being a BMO $\mathbb F$-martingale  implies clearly that $\langle K^{\mathbb F}, m \rangle ^{\mathbb F}$ has integrable variation, and hence $(1-\widetilde{G})\is \langle K^{\mathbb F}, m^{(1)} \rangle ^{\mathbb F}$ does also. Therefore, in virtue of this latter fact and Theorem \ref{existencetheorem}, we can conclude the condition  \eqref{existencecond2}  is sufficient for the existence of the log-optimal portfolio for $(S-S^{\tau}, \mathbb G)$ as soon as $({S}^{(1)}, \mathbb F)$ admits log-optimal portfolio. 
This ends the proof of the theorem. 
 \end{proof}
%%%%%%%%%%%%%%%%%%%%%%%%%%%%%%%%%%%%%%%%%%%%%%%%%%%%%%%%%%%%%%%%%%%%%%%%%Log-optimal via Predictable-characteristics%%%%%%%% %%%%%%%%%%%%%%%%%%%%%%%%%%%%%%%%%%%%%%%%%%%%%%%%%%%%%%%%%%%%%%
\section{Log-optimal portfolio for $(S-S^{\tau}, \mathbb G)$: Description and sensitivity}\label{Section4}
%%%%%%%%%%%%%%%%%%%%%%%%%%%%%%%%%%%%%%%%%%%%%%%%%%%%%%%%
This section describes explicitly the log-optimal portfolio for the model $(S-S^{\tau},\mathbb{G})$, when it exists, using the $\mathbb{F}$-parameters of the pair $(S, \tau)$. Thanks to  \cite[Theorem 2.1]{ChoulliYansori1} (see also \cite[Theorem 5.2]{ChoulliYansori2} for the case $(S^{\tau},\mathbb{G})$), this task is feasible no matter what is the initial model $(S,\mathbb{F},P)$, due to the statistical techniques called {\it the predictable characteristics} of semimartingales. For more details about these tools, we refer the reader to \cite[Chapters III and IV]{jacod79} and \cite[Section \RN{2}.2]{jacodshiryaev}, and for their applications we refer to \cite{ChoulliYansori1,ChoulliYansori2,ChoulliStricker2005,ChoulliStricker2006,ChoulliStricker2007} and the references therein to cite few.  Thus, the rest of this paragraph will parametrizes the pair $(S,\tau)$, via predictable characteristics which are $\mathbb{F}$-observable. Throughout the rest of the paper, on $\Omega \times [0, + \infty ) \times {\mathbb R}$, we consider the $\sigma$-algebras
\begin{equation*}
\widetilde{\cal O}(\mathbb F) := {\cal O}(\mathbb F) \otimes {\cal B}({\mathbb R}^d) \quad\mbox{and}\quad \widetilde{\cal P}(\mathbb F) := {\cal P}(\mathbb F) \otimes {\cal B}({\mathbb R}^d), 
\end{equation*}
where ${\cal B}({\mathbb R}^d)$ is the Borel $\sigma$-field on the optional and predictable $\sigma$-fields, respectively. \\
{\bf Parametrization of $S$ and $S^{(1)}$:} The random measure associated to the jumps of $S$, denoted by $\mu$, is given by
\begin{equation*}
    \mu(dt, dx):= \sum_{s>0} I_{\{\Delta S \neq 0\}} \delta_{(s, \Delta S_s)}(dt, dx).
\end{equation*}
For a product-measurable functional $W\geq0$ on $\Omega \times [0, + \infty ) \times {\mathbb R}$, we denote by $W \star \mu$ the process 
\begin{equation}
    W \star \mu := \int_0^t \int_{{\mathbb R} \backslash \{0\}} W
    (u, x) \mu(du, dx) = \sum_{0<u\leq t} W(u, \Delta S_u) I_{\{\Delta S \neq 0\}} 
\end{equation}
Thus, on $\Omega \times [0, + \infty ) \times {\mathbb R}$, we define the $\sigma$-finite measure $M^{\mathbb P}_{\mu} := {\mathbb P} \otimes \mu$ by 
\begin{equation*}
    \int W dM^{\mathbb P}_{\mu} := E\left ( W \star \mu_{\infty} \right )
\end{equation*}
The  $\mathbb{F}$-compensator of $\mu$ is  the random measure $\nu$ defined by $E \left ( W \star \mu_{\infty}\right ) =  E \left ( W \star \nu_{\infty} \right )$, for each $\widetilde{{\mathcal P}}(\mathbb F)$-measurable $W\geq0$. Then, by \cite[Theorem 2.34]{jacodshiryaev} and fixed truncation function $h(x):= x I_{\{\vert x \vert \leq 1\}}$ the so called "canonical representation"
of $S$ is given by the following decomposition:
\begin{equation}\label{CanonicalDecomposition4S}
    S = S_0 + S^c + h(x) \star (\mu - \nu) + b \is A +(x-h(x)) \star \nu.
\end{equation}
where $S^c$ is the continuous local martingale part of $S$. $h \star (\mu -\nu)$ is the unique pure jump ${\mathbb H}$-local martingale with jumps that given by $h(\Delta S)I_{\{\Delta S \neq 0\}} -  {^{p, \mathbb{F}}}(h(\Delta S) I_{\{\Delta S \neq 0\}})$. For $\nu$ and $C$ is the matrix with entries $C_{ij} := \langle S^{c,i}, S^{c,j} \rangle$ we can find a version satisfying 
\begin{equation*}
 C = c \is A, \quad \nu(dt, dx) = dA_t F_t(dx), \quad F_t\left(\{0\} \right)=0,  \quad \int (\vert x \vert^2 \wedge 1)F_t(dx) \leq 1.
\end{equation*}
where $A$ is increasing and continuous, $b$ and $c$ are predictable processes, $F_t(dx)$ is a predictable kernel, $b_t(w)$ is a vector in ${\mathbb R}^d$ and $c_t(w)$ is a symmetric $d \times d$-matrix, for all $(w,t) \in \Omega \times {\mathbb R}$. The quadruplet $(b, c, F, A)$ are the predictable characteristics of $S$. 
Throughout the rest of the paper, we define 
\begin{equation}
    \widehat{W}_t := \int W(t, x) \nu(\{t\}, dx), \quad a_t:= \widehat{1}_t= \nu(\{t\}, {\mathbb R}^d). 
\end{equation}
for any predictable functional $W$ such that the above integral exists. As $S^{(1)}=I_{\{G_{-}<1\}}\is{S}$, then the random measure for its jumps $\mu_1$ and its compensator random measure $\nu_1$ are given by 
\begin{equation}\label{mu1nu1}
\mu_1(dt,dx):=I_{\{G_{t-}<1\}}\mu(dt,dx)\quad\mbox{and}\quad \nu_1(dt,dx):=I_{\{G_{t-}<1\}}\nu(dt,dx).
\end{equation}
Hence, it is easy to check that the predictable characteristics $(b^{(1)}, c^{(1)}, F_1, A^{(1)})$ of $S^{(1)}$ are given by 
 \begin{equation}\label{Characteristics4S1}
  A^{(1)}:=I_{\{G_{-}<1\}}\is A,\quad F_1(t,dx):=I_{\{G_{t-}<1\}}F(t,dx),\quad b^{(1)}=b,\quad c^{(1)}=c.\end{equation}
{\bf Parametrization of $\tau$:} Thanks to \cite[Theorem 3.75]{jacod79} and \cite[Lemma 4.24]{jacodshiryaev}, we will consider Jacod's decomposition for the $\mathbb F$-martingale $m$ as follows 
\begin{equation}\label{Decomposition4m}
  {m} = \beta^{(m)} \is S^c + U^{(m)} \star (\mu-\nu) + g^{(m)}\star \mu + m^{\perp}, \quad U^{(m)} := f^{(m)} + \dfrac{\widehat{f^{(m)}}}{1-a}I_{\{a<1\}}. 
\end{equation}
For the sake of simplifying the formulas, throughout the rest of the paper, we consider the functionals
\begin{equation}\label{Proceses(m,1)}
\begin{cases}\beta^{(m,1)} :=\beta^{(m)}(1-G_{-})^{-1}I_{\{G_{-}<1\}},\quad f^{(m,1)} :=f^{(m)}(1-G_{-})^{-1}I_{\{G_{-}<1\}},\cr\\
 g^{(m,1)} :=g^{(m)}(1-G_{-})^{-1}I_{\{G_{-}<1\}},\quad m^{(\perp,1)} :=I_{\{G_{-}<1\}}(1-G_{-})^{-1}\is{m}^{\perp},\end{cases}
\end{equation}
 and the following function 
 \begin{equation}\label{klog}
   {\mathcal K}_{log}(y) := \dfrac{-y}{1+y}+\ln (1+y) \quad  \mbox{for any} \quad  y>-1.   
 \end{equation}
 The rest of this section is divided into three subsections. The first subsection elaborates our main results, and discusses the applications of these results and their financial interpretations as well. The second subsection illustrates these main results on the model where $(S,\mathbb{F})$ follows the jump-diffusion model.  The last subsection proves the main results of the first subsection.
 %%%%%%%%%%%%%%%%%%%%%%%%%%%%%%%%%%%%%%%%%%%%%%%%%%%%%%%%%%%%%%%%%%%%%%%%%%%%%%%%%% MAIN THEOREM %%%%%%%%%%%%%%%%%%%%%%%%%%%%% %%%%%%%%%%%%%%%%%%%%%%%%%%%%%%%%%%%%%%%%%%%%%%%%%%%%%%%%%%%%%%
\subsection{Main results and their applications and interpretations}
In this subsection, we start in the following theorem by describing completely and as explicit as possible the log-optimal portfolio of $(S-S^{\tau},\mathbb{G})$ using the parameters of the pair $(S,\tau)$  which are $\mathbb{F}$-observables. This allows us to single out, with sharpe precision, what are the various risks induced by $\tau$ that really affect the existence and the structure as well of the log-optimal portfolio.
 %%%%%%%%%%%%%%%%%%%%%%%%%%%%%%%%%%%%%%%%%%%%%%%%%%%%%%%%%%%%%%%%%%%%%%%%%%%%%%%%%
% \subsection{Main results}
 %%%%%%%%%%%%%%%%%%%%%%%%%%%%%%%%%%%%%%%%%%%%%%%%%%%%%%%%%%%%
 \begin{theorem}\label{generalpredictable}
  Suppose (\ref{Assumptions4Tau}) holds, and let $(\beta^{(m,1)},f^{(m,1)})$ and ${\mathcal K}_{log}$ be given by (\ref{Proceses(m,1)}) and  \eqref{klog} respectively. Then the following assertions are equivalent. \\
  {\rm{(a)}} The log-optimal portfolio $\widetilde{\theta}^{\mathbb G}$ for $(S-S^{\tau}, \mathbb G)$ exists (i.e. ${\cal D}_{log}(S-S^{\tau}, \mathbb G) \neq \phi)$. \\
  {\rm{(b)}} There exists $\widetilde{\varphi} \in {\cal L}(S^{(1)}, \mathbb F)$ such that
  \begin{equation}
    (\theta - \widetilde{\varphi})^{tr}\left\{b- c(\beta^{(m,1)}+\widetilde{\varphi})+ \int \left ( \dfrac{1-f^{(m,1)}(x)}{1+\widetilde{\varphi}^{tr}x}x- h(x)\right)F_1(dx)\right\} \leq 0,\ P\otimes A^{(1)}\mbox{-a.e.},\label{G1}\end{equation}
     for any $\theta \in {\cal L}(S^{(1)}, \mathbb F)$, and 
      \begin{equation} E \left[(1-G_{-}) \is\left( \widetilde{V}^{(1)}+\widetilde{\varphi}^{tr} c {\widetilde \varphi} \is A^{(1)}+{\mathcal K}_{log} ({\widetilde \varphi}^{tr} x)(1-f^{(m,1)})\star \nu_1 \right)_T \right ] < + \infty. \label{G2}\end{equation}
      Here
     \begin{equation}   \widetilde{V}^{(1)} := \widetilde{\varphi}^{tr} \left(b-c (\beta^{(m,1)}+ \widetilde{\varphi})\right)\is A^{(1)} + \left [ \dfrac{(1-f^{(m,1)} (x))\widetilde{\varphi}^{tr} x }{1+\widetilde{\varphi}^{tr} x} - \widetilde{\varphi}^{tr}h(x)\right]\star \nu_1. \label{V(1)}\end{equation}
Furthermore, when they exist, the processes $\widetilde{\theta}^{\mathbb G}$, $\widetilde{\varphi}$ and $\widetilde{Z}^{\mathbb G}$ and $\widetilde{Z}^{\mathbb F} \in {\cal D}({S}^{(1)}, \mathbb F)$ solution to the minimization of the RHS term of \eqref{optimalZF}  are related to each other via the following 
 \begin{align}
   & \widetilde{\theta}^{\mathbb G}  \left ( 1+ ( \widetilde{\theta}^{\mathbb G} \is (S-S^{\tau}))_{-}\right)^{-1} = \widetilde{\varphi}  \quad P\otimes A\mbox{-a.e.}  \quad \mbox{on} \quad  \Rbrack \tau , \infty \Lbrack, \label{Equation4.12}\\
   & \widetilde{Z}^{\mathbb G} = {\cal E}(\widetilde{K}^{\mathbb G}){\cal E}(-I_{\Rbrack \tau, \infty \Lbrack} \is \widetilde{V}^{(1)}) = \dfrac{{\cal E}(I_{\Rbrack \tau, \infty \Lbrack} \is \widetilde{K}^{\mathbb F}){\cal E}(-I_{\Rbrack \tau, \infty \Lbrack} \is \widetilde{V}^{(1)})}{{\cal E}(I_{\Rbrack \tau, \infty \Lbrack} \is{m}^{(1)} )} = \dfrac{\widetilde{Z}^{\mathbb F}/(\widetilde{Z}^{\mathbb F})^{\tau}}{{\cal E}(I_{\Rbrack \tau, \infty \Lbrack}\is{m}^{(1)} )},\label{Equation4.13}\\
   & \widetilde{K}^{\mathbb G}:= - \widetilde{\varphi} \is {\cal T}^{(a)}(S^c) - \dfrac{\widetilde{\Gamma}^{(1)}\widetilde{\varphi}^{tr}x}{1+ \widetilde{\varphi}^{tr}x} I_{\Rbrack \tau, \infty \Lbrack}  \star (\mu - (1-f^{(m,1)}) \is \nu), \label{Equation4.14}\\
   & \widetilde{K}^{\mathbb F} := - \widetilde{\varphi} I_{\{G_{-}<1\}}\is S^{c} - \widetilde{\Gamma}^{(1)}\is{m}^{(1)} - \dfrac{\widetilde{\Gamma }^{(1)} (1-f^{(m,1)}) \widetilde{\varphi}^{tr} x}{1+\widetilde{\varphi}x}  \star (\mu_1 -\nu_1) + \dfrac{\widetilde{\Gamma }^{(1)} (\widetilde{\varphi}^{tr} x)g^{(m,1)}}{1+\widetilde{\varphi}x} \star \mu_1, \label{KF}\\
  & \widetilde{\Gamma }^{(1)} := \left ( 1- a + \widehat{f^{(op)}} + \widehat{f^{(m,1)}} - \reallywidehat{f^{(op)}f^{(m,1)}} \right )^{-1}, \quad \quad f^{(op)}(t,x):= (1+ \widetilde{\varphi}_t^{tr}x)^{-1}. \label{Gamma(1)andf(op)}
 \end{align}
 \end{theorem}
%%%%%%%%%%%%%%%%%%%%%%%%%%%%%%%%%%%%%%%%%%%%%%%%%%%%%%%%%%%%%%%%%%%%%%%%%%%
%%%%DISCUSSIONS ABOUT THE THEOREM%%%%%%%%%%%%%%%%%%%
%%%%%%%%%%%%%%%%%%%%%%%%%%%%%%%%%%%%%%%%%%%%%%%%%%%%%%%%%%%%%%%%%%%%%%%%%%%%% 
The proof of this theorem is relegated to Subsection \ref{Subsection4Proofs} for the sake of keeping this subsection with less technicalities, while the theorem conveys two principal results that we discuss herein. The first core result, which is the equivalence between assertions (a) and (b), gives a complete and explicit characterization of the log-optimal portfolio rate. Hence, this result simultaneously gives necessary and sufficient conditions on the pair $(S,\tau)$ such that the model $(S-S^{\tau},\mathbb{G})$ admits the log-optimal portfolio. It is important to mention that this first result, as we said it in the introduction that in virtue of \cite[Theorem 2.1]{ChoulliYansori1}, also gives a complete and explicit characterization  via (\ref{G1}) for the num\'eraire portfolio rate of $(S-S^{\tau},\mathbb{G})$. This equation shows that the num\'eraire portfolio (and hence the log-optimal portfolio) is impacted by ``the correlation" between $\tau$ and $S$ only. 

The second  principal result of Theorem \ref{generalpredictable} lies in the meaning of (\ref{Equation4.12})-(\ref{Equation4.13}). On the one hand, (\ref{Equation4.12}) and the first equality of (\ref{Equation4.13}) explain the exact duality relationship between the log-optimal wealth and log-optimal deflator solution to (\ref{dualaftertau}). In fact, it is easy to check that these equalities yield $1+  \widetilde{\theta}^{\mathbb G} \is (S-S^{\tau})={\cal{E}}(\widetilde\varphi{I}_{\Rbrack \tau, \infty \Lbrack} \is{S})=1/{\widetilde{Z}^{\mathbb{G}}}$. On the other hand, the second  and third equalities in (\ref{Equation4.13}) convey the structures of the optimal deflator solution to (\ref{dualaftertau}).  In virtue of the $\mathbb{G}$-martingale decomposition in  \cite{ChoulliAlharbi} (see Theorem \ref{GeneralDefaltorDescription4afterTau}), these structures were expected, and the novelty herein resides in giving a deep description on which part of $\tau$  plays central role and how it does it.\\
%%%%%%%%%%%%%%%%%%%%%%%%%%%%%%%%%%%%%%%%%%%%%%%%%%%%%%%%%%%%%%%%%%%%%%%%%%%%%% the 

One of the important application of  Theorem \ref{generalpredictable} is the quantification of impact of $\tau$ on the maximum expected logarithm utility problem. To this end, we calculate {\it the increment in maximum expected logarithm utility from terminal wealth} for the models $(S-S^{\tau},\mathbb{G})$ and $(S,\mathbb{F})$ in the following theorem. More importantly, we single out the main factors intrinsic to $\tau$, which really measure the sensitivity of the log-optimal portfolio to $\tau$, and we quantify and interpret these factors afterwards. 
%%%%%%%%%%%%%%%%%%%%%%%%%%%%%%%%%%%%%%%%%%%%%%%%%%%%%%%%%%%%%%%%%%%%%%%%%%%%%%%

\begin{theorem}\label{riskfactors}
Suppose that  (\ref{Assumptions4Tau}) is satisfied, the log-optimal portfolio rate for $(S, \mathbb F)$ exists, and \eqref{existencecond2} holds. Then there exists $\widetilde{\varphi} \in {\cal L}(S, \mathbb F)$ satisfying \eqref{G1}, and the following equalities hold
\begin{align}
   & \Delta_T (S, \tau, \mathbb F):= u_T(S-S^{\tau}, \mathbb G) - u_T(S, \mathbb F) \nonumber \\
     & =  -  \underbrace{E \left [- ((1-\widetilde{G}) \is \widetilde{\cal H}(\mathbb G))_T +((1-\widetilde{G}) \is \widetilde{\cal H}(\mathbb F))_T - \langle \widetilde{K}^{\mathbb F} -\widetilde{L}^{\mathbb F},  I_{\{G_{-}<1\}}\is{m} \rangle_T^{\mathbb F} \right ]}_{\mbox{correlation-risk-after-$\tau$}}  \label{Equation4.17} \\
    & - \underbrace{E \left [ (\widetilde{G} \is \widetilde{\cal H}(\mathbb F))_T \right]}_{\mbox{cost-of-late-investment}} +   \underbrace{E \left[ \langle \widetilde{L}^{\mathbb F},  I_{\{G_{-}<1\}}\is{m}  \rangle_T^{\mathbb F} \right ]}_{\mbox{NP($\mathbb{F}$)-correlation}} +  \underbrace{E\left[ (1-G_{-}) \is h^{(E)} \left ( m^{(1)}, \mathbb F\right )_T \right]}_{\mbox{information-premium-after-$\tau$}}, \nonumber\\
    & = -  \underbrace{E \left [ (\widetilde{G} \is \widetilde{\cal H}(\mathbb F))_T \right]}_{\mbox{cost-of-late-investment}}+  \underbrace{E \left [\int_0^T {\cal P}_t^{(N,1)} dA^{(1)}_t \right ]}_{\mbox{num\'eraire-change-premium}} + \underbrace{E \left[ \langle \widetilde{L}^{\mathbb F}, I_{\{G_{-}<1\}}\is{m} \rangle_T^{\mathbb F} \right ]}_{\mbox{NP($\mathbb{F}$)-correlation}}.
 \label{Equation4.18}   \end{align}
Here $\widetilde{K}^{\mathbb F}$ is given by \eqref{KF}, and ${\cal P}^{(N,1)}$, $\widetilde{L}^{\mathbb F}$, $\widetilde{\cal H}(\mathbb G)$, $\widetilde{\cal H}(\mathbb F)$ are given by 

\begin{align}
  {\cal P}_t^{(N,1)}& := (1-G_{t-})\left\{(\widetilde{\varphi}_t - \widetilde{\lambda}_t)^{tr} b_t - (\widetilde{\varphi}_t - \widetilde{\lambda}_t)^{tr} c_t\beta^{(m,1)}_t - \frac{1}{2} \widetilde{\varphi}_t^{tr} c_t \widetilde{\varphi}_t + \frac{1}{2} \widetilde{\lambda}_t^{tr} c_t \widetilde{\lambda}_t \right\}\nonumber \\  
  & +(1-G_{t-}) \int \left((1-f^{(m,1)}(t, x))\ln \left( \frac{1+\widetilde{\varphi}^{tr}_t x}{1+\widetilde{\lambda}^{tr}_t x}\right) -(\widetilde{\varphi}_t - \widetilde{\lambda}_t)^{tr} h(x) \right) F_1(t,dx),\label{Equation4.19} \\
  \widetilde{L}^{\mathbb F} & := -\widetilde{\lambda} \is S^{c} - \frac{\widetilde{\Xi} \widetilde{\lambda}^{tr} x}{1+\widetilde{\lambda}^{tr}x} \star(\mu -\nu), \quad \quad \widetilde{\Xi}_t^{-1} := 1 -a_t + \int \frac{\nu({\{t\}}, dx)}{1+\widetilde{\lambda}^{tr}x}, \label{Equation4.20}  \\
  \widetilde{\cal H}(\mathbb G) & := \widetilde{V}^{(1)} + \sum \left ( - \Delta \widetilde{V}^{(1)} - \ln (1- \Delta \widetilde{V}^{(1)}) \right ) + H^{(0)} (\widetilde{K}^{\mathbb F}, \mathbb F),\label{Equation4.21}  \\
  \widetilde{\cal H}(\mathbb F) & := \widetilde{V}^{\mathbb F} + \sum \left (-\Delta \widetilde{V}^{\mathbb F} - \ln (1- \Delta \widetilde{V}^{\mathbb F}) \right ) + H^{(0)} (\widetilde{L}^{\mathbb F}, \mathbb F), \label{Equation4.22} 
\end{align}
while $\widetilde{V}^{(1)}$ and $\widetilde{V}^{\mathbb F}$ are defined by \eqref{V(1)} and 
\begin{equation}\label{Equation4.23} 
   \widetilde{V}^{\mathbb F} :=  \left (\widetilde{\lambda}^{tr} b -\widetilde{\lambda}^{tr} c\widetilde{\lambda} \right ) \is A + \left (\dfrac{\widetilde{\lambda}^{tr} x}{1+ \widetilde{\lambda}^{tr} x} - \widetilde{\lambda}^{tr} h \right) \star \nu. 
\end{equation}
\end{theorem}
%%%%%%%%%%%%%%%%%%%%%%%%%%%%%%%%%%%%%%%%%%%%%%%%%%%%%%%%%%%%%%%%
%%%%%%%%%%%%%%%%%%%%%%%%%%%%%%%%%%%%%%%%%%%%%%%%%%%%%%%%%%%%%%%%%%%%
\subsection{The case when $(S, \mathbb F)$ is a jump-diffusion model}\label{Subsection4Example}
In this section, we illustrate the results of Sections 3 and 4 on the jump-diffusion market model.  
To this end, we consider the case of one-dimensional jump-diffusion framework for the market model $(S,\mathbb F, P)$. Precisely, we suppose that on $(\Omega, {\cal{F}},P)$ a one-dimensional Brownian motion and a Poisson process $N$ with rate $\lambda>0$ are defined such that $W$ and $N$ are independent. Then $\mathbb F$ is the completed and right-continuous filtration generated by $(W,N)$, and the stock's price $S$ is given by 
 \begin{equation} \label{jumpmodel}
     S_t = S_0 +  \int_0^t S_{s-} \mu_s ds + (S_{-} \sigma  \is W)_t + (S_{-} \zeta\is N^{\mathbb F})_t, \quad \quad N^{\mathbb F}:= N - \lambda t.
 \end{equation}
Here,  ${N}^{\mathbb F}$ is the martingale compensated Poisson process, and $ \mu,\ \sigma$ , an $ \zeta $ are bounded adapted processes, and there exists a constant $\delta \in (0,+\infty)$ such that 
\begin{equation}\label{coefficientscond}
 \zeta>-1\quad\mbox{and}\quad \sigma+\vert\zeta\vert\geq \delta,\ P\otimes dt\mbox{-a.e.}.
\end{equation}
As $\mathbb{F}={\mathbb{F}}^{(W,N)}$ and $m$ is an $\mathbb{F}$-martingale, then there exists a pair $(\varphi^{(m)},\psi^{(m)})$ of $\mathbb{F}$-predictable processes such that 
\begin{equation}\label{Phi(m)Psi(m)}
m=m_0+\varphi^{(m)}\is W+\psi^{(m)}\is {N}^{\mathbb F},\quad\mbox{and}\quad \int_0^T \left((\varphi^{(m)}_s)^2+(\psi^{(m)}_s)^2\right)ds<+\infty\quad P\mbox{-a.s.}.\end{equation}
Thus, in this framework, $\tau$ is parametrized by $(\varphi^{(m)},\psi^{(m)}, G_{-})$. However, when dealing with $(S-S^{\tau},\mathbb{G})$ as it is shown previously, only the triplet $(\varphi^{(m,1)},\psi^{(m,1)}, 1-G_{-})$ defined below appears naturally.
\begin{equation}\label{Psi(m,1}
\varphi^{(m,1)}:={{\varphi^{(m)}}\over{1-G_{-}}}I_{\{G_{-}<1\}},\quad \psi^{(m,1)}:={{\psi^{(m)}}\over{1-G_{-}}}I_{\{G_{-}<1\}}.\end{equation}

\begin{theorem}\label{logportfolio-jumpmodel}
Suppose (\ref{Assumptions4Tau}) holds, $S$ is given by (\ref{jumpmodel})-(\ref{coefficientscond}), and $\mathbb F = \mathbb F^{(W, N)}$. Consider  
\begin{equation}\label{PhiTilde-LamdaTilde}
    \widetilde{\varphi}:= \dfrac {\zeta \Lambda_m + \vert \zeta \vert \sqrt{\Lambda^2_m + 4 \sigma^2 \lambda (1+\psi^{(m,1)})}-2\sigma^2}{2 \sigma^2 \zeta S_{-}} {I}_{\{G_{-}<1\}}\ \&\ 
    \widetilde{\lambda}:= \dfrac {\zeta \Lambda_0 + \vert \zeta \vert \sqrt{\Lambda^2_0 + 4 \sigma^2 \lambda}}{2 \sigma^2 \zeta S_{-}} - \dfrac{1}{\zeta S_{-}}  ,   
\end{equation}
where $\Lambda_m := \mu -\lambda \zeta - \sigma\varphi^{(m,1)} + \sigma^2\zeta^{-1}$ and $\Lambda_0:= \mu -\lambda \zeta+ \sigma^2\zeta^{-1}$.\\
Then  $\widetilde{\varphi}I_{\Rbrack \tau, \infty \Lbrack}$ is the num\'eraire portfolio rate for $(S-S^{\tau}, \mathbb G)$ and $\widetilde{\lambda}$ is the log-optimal portfolio rates for $(S, \mathbb F)$. If furthermore $\tau$ satisfies 
\begin{equation}\label{sufficientcondjump}
    E \left [\int_0^T (1-G_{s-}) \left( (\varphi^{(m,1)}_s)^2+ \lambda (1+\psi^{(m,1)}_s)\ln(1+\psi^{(m,1)}_s)-\lambda \psi^{(m,1)}_s\right) ds \right] < +\infty,
\end{equation}
then the following equivalent assertions hold.\\
{\rm{(a)}} The solution of (\ref{dualaftertau}) exists, and it is given by
\begin{equation}
    \widetilde{Z}^{\mathbb G}:= {\cal E}(\widetilde{K}^{\mathbb G}), \quad \widetilde{K}^{\mathbb G}:= - \widetilde{\varphi} \sigma \is {\cal T}^{(a)}(W) + \frac{(\psi^{(m,1)} -1) \widetilde{\varphi}\zeta S_{-}}{1+\widetilde{\varphi} \zeta S_{-}} \is {\cal T}^{(a)}(N^{\mathbb F}).
\end{equation}
{\rm{(b)}}  $\widetilde{\varphi}I_{\Rbrack \tau, \infty \Lbrack}$ is the log-optimal rate for the model $(S-S^{\tau}, \mathbb G)$.
\end{theorem}

\begin{proof} In virtue of (\ref{jumpmodel}), we deduce that $\Delta S = S_{-} \zeta \Delta N$,  and hence in this setting we calculate the pair  $(\mu, \nu)$ and the predictable characteristics quadruplet $(b,c,F, A)$ for the model $(S,\mathbb{F})$ as follows. 
\begin{equation}\label{munu}
    \mu(dx, dt):= \delta_{S_{t-} \zeta_{t}}(dx) dN_t,\quad    \nu(dx, dt):= \lambda \delta_{S_{t-} \zeta_t} (dx) dt,\end{equation}
 where $\delta_a$ is the Dirac mass at point $a$, and
  \begin{equation}\label{Characteristic4S}  
  b= (\mu - \lambda \zeta I_{\{\vert \zeta \vert S_{-} > 1  \}})S_{-}, \quad c = (\sigma S_{-})^2, \quad F_t(dx) = \lambda \delta_{\zeta_t S_{t-}}(dx),  \quad A_t = t.\end{equation}
Then, we derive $ {\cal L} (S, \mathbb F)$, which is an open set in $\mathbb R$ (when we fix $(\omega,t)\in \Omega\times[0,+\infty)$) and is
\begin{align}\label{phispace}
    {\cal L}(S, \mathbb F)  := \{ \varphi\ \mathbb{F}\mbox{-predictable} \ \vert \ \varphi S_{-} \zeta > -1 \ P\otimes dt\mbox{-a.e.}\} = ( -1 / (S_{-} \zeta )^{+}, 1/ (S_{-} \zeta )^{-}),  
\end{align}
with the convention $1/0^+= + \infty$. The Jacod's components  for $m$, i.e. $(\beta^{(m)}, f^{(m)}, g^{(m)}, m^{\perp})$ in this framework take the form of 
\begin{equation}\label{Charateristics4m}
(\beta^{(m)}, f^{(m)}, g^{(m)}, m^{\perp})=(\dfrac{\varphi^{(m)}}{\sigma S_{-}},\psi^{(m)}, 0, 0).\end{equation}
The rest of this proof is divided into two parts.\\
{\bf Part 1.} This part proves that $ \widetilde{\lambda}$ is the log-optimal portfolio rate of $(S,\mathbb{F})$.\\
Thanks to \cite[Theorem 2.1]{ChoulliYansori1}, the log-optimal portfolio rate is the unique solution $\widetilde{\psi}\in  {\cal L}(S, \mathbb F) $ to 
\begin{equation}\label{logOptimalEquation4S}
 (\psi - \widetilde{\psi})^{tr}\left\{b- c\widetilde{\psi}+ \int \left ( (1+\widetilde{\psi}^{tr}x)^{-1}x- h(x)\right)F(dx)\right\} \leq 0,\ P\otimes A\mbox{-a.e.},\end{equation}
     for any $\psi \in {\cal L}(S, \mathbb F)$, and 
      \begin{equation} \label{Integrability4PsiTilde}
      E \left[(1-G_{-}) \is\left( \widetilde{V}^{\mathbb{F}}+{\widetilde\psi}^{tr} c {\widetilde\psi} \is A+{\mathcal K}_{log} ({\widetilde\psi}^{tr} x)\star \nu\right)_T \right ] < + \infty,\end{equation}
     where ${\mathcal K}_{log}$ is given by  (\ref{klog}) and $\widetilde{V}^{\mathbb{F}}$ is given by 
     \begin{equation} \label{Vtilde(F)}
     \widetilde{V}^{\mathbb{F}} := \widetilde{\psi}^{tr} \left(b-c\widetilde\psi\right)\is A + \left( {{{\widetilde\psi}^{tr} x }\over{1+\widetilde{\psi}^{tr} x}} - \widetilde{\psi}^{tr}h(x)\right)\star \nu.\end{equation}
Thus, by inserting (\ref{Characteristic4S} ) in (\ref{logOptimalEquation4S}),  and using (\ref{phispace}) which claims that $ {\cal L}(S, \mathbb F) $ is an open set of $\mathbb{R}$ point-wise in $(\omega,t)$, we deduce that $\widetilde{\psi}\in  {\cal L}(S, \mathbb F)$  is the unique solution to
\begin{align*}
    0 = \mu -\lambda \zeta - \widetilde{\psi} S_{-} \sigma^2 + \frac{ \lambda \zeta }{1+ \widetilde{\psi} \zeta S_{-}}.
\end{align*}
It is clear that $\widetilde\lambda$ given in (\ref{PhiTilde-LamdaTilde}) is the unique solutions to the above equation belonging to $ {\cal L}(S, \mathbb F) $, and hence $\widetilde\lambda$ is the num\'eraire portfolio rate for $(S,\mathbb{F})$. To prove that $\widetilde\lambda$ is the log-optimal portfolio rate, we need to check that (\ref{Integrability4PsiTilde}) holds. To this end, it is easy to see that in our setting,  (\ref{Integrability4PsiTilde})  becomes 
\begin{equation}
   E \left [ \int_0^T (1-G_{t-})\left \{ (\widetilde{\lambda}_t \sigma_t S_{t-})^2 + \lambda {\cal K}_{log}(\widetilde{\lambda}_t S_{t-} \zeta_t)   \right\} dt \right ] < +\infty .
\end{equation}
This latter condition is always true due to the fact that the processes $\widetilde{\lambda}  S_{-}$, $\sigma$ and
\begin{equation*}
   (1+ \widetilde{\lambda} S_{-} \zeta)^{-1} = \frac{-\zeta \Lambda + \sqrt{(\zeta \Lambda)^2 +4 \sigma^2 \lambda \zeta^2}}{2 \lambda \zeta^2},
\end{equation*}
are bounded due to \eqref{coefficientscond}. This proves that indeed $\widetilde{\lambda}$ is the log-optimal portfolio rate of $(S,\mathbb{F})$.\\
%%%%%%%%%%%%%%%%%%%%%%%%%%%%%%%%%%%%%%%%%%%%%%%%%%%%%%%%%%%%%%%%%%%%%%%
{\bf Part 2.} Here we prove that $ \widetilde{\varphi}I_{\Rbrack \tau, \infty \Lbrack}$ is the num\'eraire portfolio rate for $(S-S^{\tau},\mathbb{G})$, or equivalently  $ \widetilde{\varphi}$  is the unique solution to (\ref{G1}). Then by plugging (\ref{Characteristic4S}) and (\ref{Charateristics4m}) in (\ref{G1}),  and using the fact that $ {\cal L}(S, \mathbb F) $ is an open set of $\mathbb{R}$ point-wise in $(\omega,t)$ due to (\ref{phispace}) and $ {\cal L}(S^{(1)}, \mathbb F)=\left\{\varphi\ \mathbb{F}\mbox{-predictable}\ :\ \varphi{I}_{\{G_{-}<1\}}\in{\cal L}(S, \mathbb F)\right\}$, we deduce that on $\widetilde{\varphi}$ is the unique solution to 
\begin{align*}
    0 = \left(\mu -\lambda \zeta - \varphi^{(m,1)} \sigma - \widetilde{\varphi} S_{-} \sigma^2 + \frac{( 1-\psi^{(m,1)}) \lambda \zeta }{1+ \widetilde{\varphi} \zeta S_{-}}\right)I_{\{G_{-}<1\}}.
\end{align*}
Thus, direct calculation shows that $\widetilde{\varphi}$ is the one given by (\ref{logportfolio-jumpmodel}), and this proves that  $ \widetilde{\varphi}I_{\Rbrack \tau, \infty \Lbrack}$ is the num\'eraire portfolio rate of $(S-S^{\tau},\mathbb{G})$, and the second par is complete.\\
%%%%%%%%%%%%%%%%%%%%%%%%%%%%%%%%%%%%%%%%%%%%%%%%%%%%%%%%%%%%%%%%%
{\bf Part 3.}  This part proves that the two assertion (a) and (b) are equivalent and hold under (\ref{sufficientcondjump}). To this end,  we remark that due to (\ref{Phi(m)Psi(m)}), (\ref{m(a)}) and Definition \ref{hellinger}, we calculate 
\begin{align*}
    h^{(E)}(m^{(1)}, \mathbb F) & = \frac{1}{2} \int^{\cdot}_0(\varphi^{(m,1)}_s)^2 ds + \left ( \sum_{0<s\leq\cdot} (1+\psi^{(m,1)}_s\Delta{N}_s)\ln(1+\psi^{(m,1)}_s\Delta{N}_s)- \psi^{(m,1)}_s\Delta{N}_s \right )^{p, \mathbb F}\\
    & = \frac{1}{2} \int^{\cdot}_0(\varphi^{(m,1)}_s)^2 ds + \left (\sum_{0<s\leq\cdot} (1+\psi^{(m,1)}_s)\ln(1+\psi^{(m,1)}_s) \Delta{N}_s- \psi^{(m,1)}_s\Delta{N}_s \right )^{p, \mathbb F}\\
    & = \int_0^T \left ( \frac{1}{2} (\varphi^{(m,1)}_s)^2 + \lambda (1+\psi^{(m,1)}_s)\ln(1+\psi^{(m,1)}_s) - \lambda \psi^{(m,1)}_s \right) ds.
\end{align*}
Thus, the condition (\ref{sufficientcondjump}) is the version of the condition (\ref{existencecond2}) in the current setting of jump-diffusion. Hence, by combining part 1 (assertion (a)) and Theorem \ref{sufficientcond}, we deduce that assertion (b) holds. The equivalence between assertion (a) and (b) is a direct consequence of Theorem \ref{generalpredictable} combined with part 2, and both (\ref{Characteristic4S} ) and (\ref{Charateristics4m}) which imply 
\begin{align*}
    \widetilde{K}^{\mathbb G} = & - \widetilde{\varphi} \is {\cal T}^{(a)}(S^c) - \dfrac{\widetilde{\varphi} x}{1+ \widetilde{\varphi}x} I_{\Rbrack \tau, \infty \Lbrack}  \star (\mu - (1-f^{(m,1)}) \is \nu) \\
    = & - \widetilde{\varphi} \sigma{S}_{-}\is {\cal T}^{(a)}(W) - \frac{ \widetilde{\varphi}\zeta S_{-}}{1+\widetilde{\varphi} \zeta S_{-}} \is({I}_{\Rbrack \tau, \infty \Lbrack}\is{N}-\lambda(1-\psi^{(m,1)})\is{dt}) \\
    = & - \widetilde{\varphi} \sigma{S}_{-}\is {\cal T}^{(a)}(W) - \frac{(1- \psi^{(m,1)} ) \widetilde{\varphi}\zeta S_{-}}{1+\widetilde{\varphi} \zeta S_{-}} \is {\cal T}^{(a)}(N^{\mathbb F}) 
.
\end{align*}
The last equality is a direct consequence of $(1-\psi^{(m,1)})d{\cal{T}}^{(a)}(N^{\mathbb{F}})={I}_{\Rbrack \tau, \infty \Lbrack}d{N}-\lambda(1-\psi^{(m,1)}){I}_{\Rbrack \tau, \infty \Lbrack}{dt}$. This ends the proof of the theorem. 
\end{proof}
%%%%%%%%%%%%%%%%%%%%%%%%%%%%%%%%%%%%%%%%%%%%%%%%%%%%%%%%%%%%%%%%%%%%%%%%%%%%%%%%%%%%%%%%%
%%%%%%%%%%%%%%%%%%%%%%%%%%%%%%%%%%%%%%%%%%%%%%%%%%%%%%%%%%%%%%%%%%%%%%%%%%%%%%%%%%%%%%%%
 \subsection{Proof of Theorems \ref{generalpredictable} and  \ref{riskfactors}}\label{Subsection4Proofs}
 This subsection details the proof of Theorems  \ref{generalpredictable} and  \ref{riskfactors} and elaborates the intermediate technical lemmas that are vital for the proof of these theorems. The proof of Theorems \ref{generalpredictable}  is based essentially on three intermediate lemmas that are interesting in themselves. The first lemma determines the triplet $(\mu^{\mathbb{G}},\nu^{\mathbb{G}},S^{c,\mathbb{G}})$, which is constituted by the random measure of the jumps of $S-S^{\tau}$, its $\mathbb{G}$-compensator random measure, and the continuous $\mathbb{G}$-local martingale part of $S-S^{\tau}$.  
 %%%%%%%%%%%%%%%%%%%%%%%%%%%%%%%%%%%%%%%%%%%%%%%%%%%%%%%%%%%%%%%%%%%%%%
%%%%%%%%%%%%%%%%%%%%%%%%%%%%%%%%%%%%%%%%%%%%%%%%%%%%%%%%%%%%%%%
\begin{lemma}\label{PredictableCharateristoics4S(tau)} The following assertions hold.\\
{\rm{(a)}}The random measure of the jumps of $(S-S^{\tau},\mathbb{G})$ is given by 
\begin{equation}\label{Grandommeasure}
    \mu^{\mathbb G}(dt, dx):= I_{\{t> \tau\}} \mu(dt, dx),
\end{equation}
and hence its $\mathbb{G}$-compensator random measure, denoted by $\nu^{\mathbb{G}}$, is given by 
\begin{equation}\label{Gmeasurecompensator}
    \nu^{\mathbb G}(dt, dx):= I_{\{t>\tau\}} (1-f^{(m,1)}(x,t))\nu(dt, dx).
\end{equation}
{\rm{(b)}} The continuous $\mathbb{G}$-local martingale part of $(S-S^{\tau},\mathbb{G})$ is given by 
\begin{equation}\label{Scontinuousaftertau}
 S^{c, \mathbb G}={\cal{T}}^{(a)}(S^c)=   I_{\Rbrack \tau, \infty \Lbrack} \is S^{c} + c \beta^{(m,1)}I_{\Rbrack \tau, \infty \Lbrack} \is A .
\end{equation}
\end{lemma}
 %%%%%%%%%%%%%%%%%%%%%%%%%%%%%%%%%%%%%%%%%%%%%%%%%%%%%%%%%%%%%%%%%%%%%%%
 The proof of the lemma is relegated to Appendix \ref{proof4lemmas}. Our second lemma connects some $\mathbb{G}$-integrability of $\mathbb{F}$-integrability using the random measures. 
\begin{lemma}\label{Existence4KFtilde} The following assertions hold.\\
{\rm{(a)}} The process $\widetilde{\Gamma}^{(1)}$, defined in (\ref{Gamma(1)andf(op)}) is positive, and both $\widetilde{\Gamma}^{(1)}$  and $1/\widetilde{\Gamma}^{(1)}$ are $\mathbb{F}$-locally bounded. \\
{\rm{(a)}} Let $f$ and $g$ be functionals that are $\widetilde{\cal{P}}(\mathbb{F})$-measurable  and $\widetilde{\cal{O}}(\mathbb{F})$-measurable respectively. If $f\star(\mu^{\mathbb{G}}-\nu^{\mathbb{G}})$ and $g\star\mu^{\mathbb{G}}$ are well defined $\mathbb{G}$-local martingales, then $f(1-f^{(m,1)})\star(\mu_1-\nu_1)$ and $g(1-f^{(m,1)})\star\mu_1$ are well defined $\mathbb{F}$-local martingale. 
\end{lemma}
The proof of the lemma mimics the proof of \cite[Lemma 5.12]{ChoulliYansori2}, and it will be omitted. Below, we elaborate our  third lemma, which connects positive $\mathbb{G}$-supermartingale to $\mathbb{F}$-supermartingales.
%%%%%%%%%%%%%%%%%%%%%%%%%%%%%%%%%%%%%%%%%%%%%%%%%%%%%%%%%%%%%%%%%%%%%%%
 \begin{lemma}\label{F-supermartingale2G-supermartingale}
 Let $X$ be a ${\mathbb F}$-semimartingale such that 
 $$\Delta X> -1\quad\mbox{and}\quad I_{\{G_{-}=1\}}\is X=0.$$
 Then ${\cal E}(X)$ is a nonnegative ${\mathbb F}$- supermartingale if and only if 
 \begin{equation}
     Y:= \dfrac{{\cal E}(I_{\Rbrack \tau , \infty \Lbrack} \is X)}{{\cal E} (-(1-G_{-})^{-1} I_{\Rbrack \tau , \infty \Lbrack} \is m )} \quad \quad \quad \mbox{is a} \quad  {\mathbb G}\mbox{-supermartingale}. 
 \end{equation}
 \end{lemma}
 %%%%%%%%%%%%%%%%%%%%%%%%%%%%%%%%%%%%%%%%%%%%%%%%%%%%%%%%%%%%%%%%%%%%%%
%%%%%%%%%%%%%%%%%%%%%%%%%%%%%%%%%%%%%%%%%%%%%%%%%%%%%%%%%%%%%%%%%%%%%%%%%%%
The proof of this lemma is relegated to Appendix \ref{proof4lemmas}, while below we prove Theorem \ref{generalpredictable}.
 %%%%%%%%%%%%%%%%%%%%%%%%%%%%%%%%%%%%%%%%%%%%%%%%%%%%%%%%%%%%%%%%%%%%%%%%%%%%
%%%%%%%%%%%%%%%%%%%%%%%%%%%%%%%%%%%%%%%%%%%%%%%%%%%%%%%%%%%%%%%%%%%%%%%%%%%
 \begin{proof}[Proof of Theorem \ref{generalpredictable}] The proof of  this theorem requires the predictable characteristics of the model $(S-S^{\tau},\mathbb{G})$, and hence we divided this proof into two parts. The first part derives the predictable characteristics of $(S-S^{\tau},\mathbb{G})$, while the second part proves the statements of the theorem.\\
 {\bf Part 1.} This part discusses the predictable characteristics of $(S-S^{\tau},\mathbb{G})$. Thus, by combining Lemma \ref{PredictableCharateristoics4S(tau)} and  (\ref{CanonicalDecomposition4S}), we derive
\begin{align}
S -S^{\tau} & = I_{\Rbrack \tau, \infty \Lbrack} \is S=  I_{\Rbrack \tau, \infty \Lbrack} \is S^{c} + I_{\Rbrack \tau, \infty \Lbrack} h \star (\mu - \nu) +I_{\Rbrack \tau, \infty \Lbrack} b \is A + I_{\Rbrack \tau, \infty \Lbrack} (x-h) \star \mu  \nonumber\\
& = S^{c, \mathbb G} +h \star(\mu^{\mathbb G}-\nu^{\mathbb G})+  I_{\Rbrack \tau, \infty \Lbrack} b \is A - c\beta^{(m,1)} I_{\Rbrack \tau, \infty \Lbrack} \is A - \left ( \int h(x) f^{(m,1)}(x) F(dx) \right ) I_{\Rbrack \tau, \infty \Lbrack} \is A  \nonumber\\ 
 & \hskip 1cm +(x-h)\star\mu^{\mathbb G} \nonumber\\ 
 & =  S^{c, \mathbb G} +h \star(\mu^{\mathbb G}-\nu^{\mathbb G}) + \left( b- c\beta^{(m,1)} -\int h f^{(m,1)} F(dx)\right ) I_{\Rbrack\tau,+\infty\Lbrack} \is A +(x-h)\star\mu^{\mathbb G} \nonumber.   
\end{align}
Hence, this $\mathbb{G}$-canonical decompositiojn of $(S-S^{\tau},\mathbb{G})$ allows us to obtain the predictable characteristics $\left(b^{\mathbb G},c^{\mathbb G},F^{\mathbb G}, A^{\mathbb G}\right)$ of $(S-S^\tau,\mathbb G)$ as follows 
\begin{equation}\label{predictablechar}\begin{cases}
 b^{\mathbb G}: = \Bigl(b-  c\beta^{(m,1)}- \int h(x) f^{(m,1)}(x) F(dx)\Bigr) I_{\Rbrack\tau,+\infty\Lbrack}, \quad \quad c^{\mathbb G}:=  I_{\Rbrack\tau,+\infty\Lbrack}c, \\ 
F^{\mathbb G}(dx):=  I_{\Rbrack\tau,+\infty\Lbrack}\left( 1- f^{(m,1)}(x)\right)F(dx), \quad \quad A^{\mathbb G}:= I_{\Rbrack\tau,+\infty\Lbrack}\is{A}.\end{cases}
\end{equation}
%%%%%%%%%%%%%%%%%%%%%%%%%%%%%%%%%%%%%%%%%%%%PART2%%%%%%%%%%%%%%%%%%%%%%%%%%%%
 {\bf Part 2:} By applying \cite[Theorem 2.1]{ChoulliYansori3} directly to $(S-S^{\tau}, {\mathbb{G}})$, we deduce that assertion (a) holds if and only if there exists $\widetilde{\varphi}^{\mathbb G} \in {\cal L}(S-S^{\tau}, \mathbb G)$ such that for any $\varphi \in {\cal L}_b(S-S^{\tau}, \mathbb G)$ the following hold:
 \begin{equation}\label{Equation1}
  (\varphi-\widetilde{\varphi}^{\mathbb G})^{tr}(b^{\mathbb G}- c^{\mathbb G}\widetilde{\varphi}^{\mathbb G}) + \int \left ( \dfrac{(\varphi-\widetilde{\varphi}^{\mathbb G})^{tr}x}{1+(\widetilde{\varphi}^{\mathbb G})^{tr}x} - (\varphi-\widetilde{\varphi}^{\mathbb G})^{tr} h(x) \right ) F^{\mathbb G}(dx) \leq 0,
   \end{equation}
   and  
    \begin{equation}\label{Equation2}
    E \left [ \widetilde{V}_T^{\mathbb G} +\frac{1}{2} ((\widetilde{\varphi}^{\mathbb G})^{tr} c^{\mathbb G} \widetilde{\varphi}^{\mathbb G}\is A^{\mathbb G} )_T + ( {\mathcal K}_{log} ((\widetilde{\varphi}^{\mathbb G})^{tr} x) \star \nu^{\mathbb G} )_T \right ] < + \infty,  \end{equation}
    where 
       \begin{equation}\label{V(G)} 
       \widetilde{V}^{\mathbb G} := \left[ (\widetilde{\varphi}^{\mathbb G})^{tr} (b^{\mathbb G} - c^{\mathbb G} \widetilde{\varphi}^{\mathbb G}) + \int \left ( \dfrac{(\widetilde{\varphi}^{\mathbb G})^{tr}x}{1+(\widetilde{\varphi}^{\mathbb G})^{tr} x} -(\widetilde{\varphi}^{\mathbb G})^{tr} h(x) \right ) F^{\mathbb G}(dx) \right]\is A^{\mathbb G}. 
 \end{equation}
 Furthermore, the following properties hold
 \begin{align}
  & \widetilde{\theta}^{\mathbb G}  \left ( 1+ ( \widetilde{\theta}^{\mathbb G} \is (S-S^{\tau}))_{-}\right)^{-1} =\widetilde{\varphi}^{\mathbb G}  \quad P\otimes A\mbox{-a.e.}  \quad \mbox{on} \quad  \Rbrack \tau , \infty \Lbrack, \label{ThetaG2Phitilde}\\
   & {Z}^{\mathbb G} = {\cal E}({K}^{\mathbb G}){\cal E}(-I_{\Rbrack \tau, \infty \Lbrack} \is \widetilde{V}^{\mathbb{G}} )={1\over{{\cal{E}}( \widetilde{\varphi}^{\mathbb G}\is(S-S^{\tau}))}},\label{Equation4.49}\\
   &{K}^{\mathbb G}:= -\widetilde{\varphi}^{\mathbb G}\is {\cal T}^{(a)}(S^c) - \dfrac{\widetilde{\Gamma}^{\mathbb{G}}(\widetilde{\varphi}^{\mathbb G})^{tr}x}{1+ (\widetilde{\varphi}^{\mathbb G})^{tr}x} \star (\mu^{\mathbb{G}} - \nu^{\mathbb{G}} ), \label{KG2PhiTilde}\\
%%%   & \widetilde{K}^{\mathbb F} := - \widetilde{\varphi} I_{\{G_{-}<1\}}\is S^{c} - \widetilde{\Gamma}^{(1)}\is{m}^{(1)} - \dfrac{\widetilde{\Gamma }^{(1)} (1-f^{(m,1)}) \widetilde{\varphi}^{tr} x}{1+\widetilde{\varphi}x}  \star (\mu_1 -\nu_1) + \dfrac{\widetilde{\Gamma }^{(1)} \widetilde{\varphi}^{tr} x}{1+\widetilde{\varphi}x} \star \mu_1, \label{KF}\\
  & \widetilde{\Gamma }^{\mathbb{G}}_t := \left ( 1 + \int (f^{(op,\mathbb{G})}(t,x)-1)\nu^{\mathbb{G}}(\{t\},dx) \right )^{-1}, \quad \quad f^{(op,\mathbb{G})}(t,x):=(1+ (\widetilde{\varphi}^{\mathbb G}_t)^{tr}x)^{-1}. \label{GammaG}
\end{align}
 Then, in virtue of Lemma \ref{Theta(G)2Theta(F)}-(c),  we obtain the existence of $\widetilde{\varphi}\in   {\cal L}(S^{(1)}, \mathbb{F})$ such that  
  \begin{equation}\label{phi}
 \widetilde{\varphi}^{\mathbb G}{I}_{\Rbrack \tau, \infty \Lbrack} = \widetilde{\varphi} I_{\Rbrack \tau , \infty \Lbrack}.  
 \end{equation}
  Thus, by inserting this latter equality with (\ref{Gmeasurecompensator}) in (\ref{GammaG}), and using (\ref{Gamma(1)andf(op)}), we obtain
  \begin{equation}\label{Gammag2Gamm(1)}
   f^{(op,\mathbb{G})}= f^{(op)}\quad\mbox{and}\quad  \widetilde{\Gamma }^{\mathbb{G}}= \widetilde{\Gamma }^{(1)}\quad\mbox{on}\quad \Rbrack \tau, \infty \Lbrack,
  \end{equation}
 Similarly, by inserting (\ref{phi}) and \eqref{predictablechar}  in (\ref{V(G)})  and using (\ref{V(1)}) afterwards, we get 
  \begin{equation}\label{V(G)bis} 
       V^{\mathbb G} =I_{\Rbrack\tau,+\infty\Lbrack}\is{V}^{(1)}. 
 \end{equation}
  As a consequence, by combining this equality with (\ref{phi}), (\ref{Gammag2Gamm(1)})  and Lemma \ref{PredictableCharateristoics4S(tau)}-(a), we get
   \begin{equation}\label{KG2KtildeG}
   {K}^{\mathbb G}=\widetilde{K}^{\mathbb G}\quad\mbox{ and}\quad {Z}^{\mathbb G} = \widetilde{Z}^{\mathbb G}=1/{\cal{E}}(\widetilde{\varphi}\is(S-S^{\tau})). \end{equation}
Again, by plugging (\ref{phi}) and (\ref{predictablechar}) in (\ref{Equation1})-(\ref{Equation2}), we conclude that assertion (a) holds if and only if there exists  $\widetilde{\varphi}\in   {\cal L}(S^{(1)}, \mathbb{F})$ such that for any  ${\varphi}\in   {\cal L}_b(S^{(1)}, \mathbb{F})$, on $\Rbrack\tau,+\infty\Lbrack$ $P\otimes A$-a.e. 
 \begin{equation}\label{Equation1bis}
  (\theta - \widetilde{\varphi})^{tr}( b+ c\beta^{(m)}- c \widetilde{\varphi}) + \int (\theta- \widetilde{\varphi})^{tr}\left ( \dfrac{ (1-f^{(m,1)}(x)x}{1+ \widetilde{\varphi}^{tr}x} - h(x) \right )F(dx) \leq 0,
   \end{equation}
   and  (\ref{G2}) holds. Thus the proof of (a) $\Longleftrightarrow$ (b) follows immediately from combining these facts with Lemma \ref{VG2VF}-(c). \\
   {\bf Part 3.} This part focuses on proving (\ref{Equation4.12}), (\ref{Equation4.13}) and (\ref{Equation4.14}).\\
   Remark that (\ref{Equation4.12}) is a direct consequence of  (\ref{phi})  and (\ref{ThetaG2Phitilde}), while both (\ref{Equation4.14}) and the first equality in (\ref{Equation4.13}) follow combining (\ref{KG2KtildeG}), (\ref{KG2PhiTilde}), (\ref{phi}), (\ref{Gammag2Gamm(1)}), (\ref{Equation4.49}) and (\ref{KG2PhiTilde}). Thus, the rest of this part focuses on proving the second equality in (\ref{Equation4.13}), or equivalently in virtue of (\ref{V(G)bis})
   \begin{equation}\label{KGtilde2KFtilde} \widetilde{K}^{\mathbb{G}} = {\cal T}^{(a)}( \widetilde{K}^{\mathbb{F}})+ (1-G_{-})^{-1} I_{\Rbrack \tau , \infty \Lbrack }\is{\cal T}^{(a)} (m).\end{equation}
To this end we use the fact that two local martingales are equal if and only if their continuous local martingale parts coincides and their jump parts are equal also. On the one hand, we use the fact that 
$\Delta {\cal T}^{(a)}(X) = \dfrac{1-G_{-}}{1-\widetilde{G}} I_{\Rbrack \tau , \infty \Lbrack} \Delta X $, and derive 
\begin{align*}
\Delta {\cal T}^{(a)}(I_{\Rbrack \tau , \infty \Lbrack} (K^{\mathbb F} + (1-G_{-})^{-1} \is m )) & = \frac{1-G_{-}}{1-\widetilde{G}} I_{\Rbrack \tau, \infty \Lbrack} \left (\Delta K^{\mathbb F} - \frac{1-\widetilde{G}}{1-G_{-}}-1 \right )\\
& = \left ( \frac{1-G_{-}}{1-\widetilde{G}} \left(\Delta K^{\mathbb F} -1\right) +1 \right ) I_{\Rbrack \tau, \infty \Lbrack}.
\end{align*}
On the other hand, thanks to (\ref{Equation4.14}), we calculate  
\begin{align*}
    \Delta {\widetilde{K}^{\mathbb G}}     & =  I_{\Rbrack \tau , \infty \Lbrack} \left(- \frac{\widetilde{\Gamma}^{(1)} \widetilde{\varphi}^{tr} \Delta S}{1+\widetilde{\varphi}^{tr}\Delta S} +  \int \frac{\widetilde{\Gamma}^{(1)} \widetilde{\varphi}^{tr} x}{1+\widetilde{\varphi}^{tr}x} F^{\mathbb G} (dx)dA_t^{\mathbb G}\right)\\
    & = I_{\Rbrack \tau , \infty \Lbrack} \left(- \frac{\widetilde{\Gamma}^{(1)} \widetilde{\varphi}^{tr} \Delta S}{1+\widetilde{\varphi}^{tr}\Delta S} + \widetilde{\Gamma}^{(1)} \int (1- \frac{1}{1+\widetilde{\varphi}^{tr}x}) \nu^{\mathbb G}(\{t\},dx)\right)\\
    & = I_{\Rbrack \tau , \infty \Lbrack} \left(- \frac{\widetilde{\Gamma}^{(1)} \widetilde{\varphi}^{tr}\Delta S}{1+\widetilde{\varphi}^{tr}\Delta S} + \widetilde{\Gamma}^{(1)} \int (1-f^{(op)})(1-f^{(m)}) \nu(\{t\},dx)\right) \\
    & = I_{\Rbrack \tau , \infty \Lbrack} \left(- \frac{\widetilde{\Gamma}^{(1)}\widetilde{\varphi}^{tr} \Delta S}{1+\widetilde{\varphi}^{tr}\Delta S} + \widetilde{\Gamma}^{(1)} \left (a- \widehat{f^{(op)}}- \widehat{f^{(m)}}+\widehat{f^{(op)}f^{(m)}}\right)\right)\\
    \Delta \widetilde{K}^{\mathbb G} & = I_{\Rbrack \tau, \infty \Lbrack} \left ( \frac{\widetilde{\Gamma}^{(1)}}{1+\widetilde{\varphi}^{tr}\Delta S} -1 \right )
\end{align*}
Then, by comparing the jump parts and using $^{o, \mathbb F} \left( I_{\Rbrack \tau, \infty \Lbrack} \right) = (1- \widetilde{G}) I_{\Rbrack 0, +\infty \Lbrack}$, we get 
\begin{align}
     \Delta \widetilde{K}^{\mathbb F} & = \dfrac{(1-\widetilde{G}) \widetilde{\Gamma}^{(1)}}{(1-G_{-})(1+\widetilde{\varphi}^{tr}\Delta S)} -1 
\end{align}
Thus, by combining $(\Delta \widetilde{K}^{\mathbb F} = \Delta \widetilde{K}^{\mathbb F}I_{\{\Delta S \neq 0 \}}+\Delta \widetilde{K}^{\mathbb F}I_{\{\Delta S \neq 0 \}})$ with the fact that on $\{\Delta S \neq 0 \}, 1- \widetilde{G} = (1-G_{-}) \left( 1- f^{(m,1)}(\Delta S) - g^{(m,1)}(\Delta S)\right )$, we obtain on $(G_{-}<1)$,
\begin{align}
      \Delta \widetilde{K}^{\mathbb F} = & \left(\dfrac{( 1- \widetilde{G}) \widetilde{\Gamma}^{(1)}}{(1+\widetilde{\varphi}^{tr}\Delta S)} -1\right)I_{\{\Delta{S}\not=0\}}  + \left (\dfrac{(1-\widetilde{G}) \widetilde{\Gamma}^{(1)}}{1-G_{-}} -1\right)I_{\{\Delta S = 0 \}}\nonumber\\
      = & - \dfrac{( 1- \widetilde{G}) \widetilde{\Gamma}^{(1)}\widetilde{\varphi}^{tr}\Delta S}{(1+\widetilde{\varphi}^{tr}\Delta S)} + \dfrac{(1-\widetilde{G}) \widetilde{\Gamma}^{(1)}}{(1-G_{-})} -1\nonumber\\
      =  &  \dfrac{ \widetilde{\Gamma}^{(1)} (f^{(m,1)}(\Delta S)-1)\widetilde{\varphi}^{tr} \Delta S} {1+\widetilde{\varphi}^{tr}\Delta S} + \dfrac{\widetilde{\Gamma}^{(1)}g^{(m,1)}(\Delta S)\widetilde{\varphi}^{tr}\Delta S}{1+\widetilde{\varphi}^{tr}\Delta S} -\widetilde{\Gamma}^{(1)} \Delta{m}^{(1)}+ \widetilde{\Gamma}^{(1)} -1 \label{LastEquality}
\end{align}
Put
$$W^{(1)} (t,x):= \widetilde{\Gamma}^{(1)}_t (f^{(m,1)}(t,x)-1)\left( \frac{1}{1+\widetilde{\varphi}^{tr}_t  x} -1\right) =   \frac{\widetilde{\Gamma}^{(1)}_t \widetilde{\varphi}^{tr}_tx(f^{(m,1)}(t,x)-1)}{1+\widetilde{\varphi}^{tr}_t  x},$$ and thanks to Lemma \ref{Existence4KFtilde}, we deduce that $W^{(1)}\star(\mu_1-\nu_1)$ and $ \dfrac{\widetilde{\Gamma}^{(1)}g^{(m,1)}\widetilde{\varphi}^{tr}x}{1+\widetilde{\varphi}^{tr}x}\star\mu_1$  are well defined $\mathbb{F}$-local martingales, and furthermore 
$$\widehat{W^{(1)}} =  1 - \widetilde{\Gamma}^{(1)}\quad\mbox{and}\quad \Delta\left({W}^{(1)}\star(\mu_1-\nu_1)\right)=\dfrac{ \widetilde{\Gamma}^{(1)} (f^{(m,1)}(\Delta S)-1)\widetilde{\varphi}^{tr} \Delta S} {1+\widetilde{\varphi}^{tr}\Delta S} + \widetilde{\Gamma}^{(1)} -1 .$$
Thus, by combining these remarks with (\ref{LastEquality}) and Lemma \ref{Existence4KFtilde}, we deduce that 
$$
  \Delta \widetilde{K}^{\mathbb F} = \Delta\left({W}^{(1)}\star(\mu_1-\nu_1)\right)+\Delta\left( \dfrac{\widetilde{\Gamma}^{(1)}g^{(m,1)}\widetilde{\varphi}^{tr}x}{1+\widetilde{\varphi}^{tr}x}\star\mu_1\right)-\widetilde{\Gamma}^{(1)} \Delta{m}^{(1)}.$$
  Therefore, this latter equality combined with the continuous $\mathbb{G}$-local martingale of $\widetilde{K}^{\mathbb{G}}$ and the continuous $\mathbb{F}$-local martingale part of $ \widetilde{K}^{\mathbb F}$ are 
  $$( \widetilde{K}^{\mathbb F})^c:=-\widetilde\varphi\is S^c-(1-G_{-})^{-1}I_{\{G_{-}<1\}}\is m^c\quad\mbox{and}\quad  ( \widetilde{K}^{\mathbb{G}})^c:=-\widetilde\varphi\is {\cal{T}}^{(a)}(S^c).$$ Thus, these yield that (\ref{KGtilde2KFtilde}) holds with $\widetilde{K}^{\mathbb F} $ is given by (\ref{KF}). This proves the second equality in (\ref{Equation4.13}), and the proof of theorem will be complete as soon as we  show that $ {\cal E }( \widetilde{K}^{\mathbb F}) {\cal E }(-\widetilde{V}^{(1)}) \in {\cal D}({S}^{(1)}, \mathbb F)$, equivalently  $ {\cal E }( \widetilde{K}^{\mathbb F}) {\cal E }(- \widetilde{V}^{(1)}){\cal E}(\psi \is{S}^{(1)})$ is an $\mathbb F$-supermartinagale for any $\psi\in {\cal{L}}({S}^{(1)}, \mathbb F)\cap{L}({S}^{(1)}, \mathbb F)$. To this end, we consider $\varphi\in {\cal{L}}({S}^{(1)}, \mathbb F)\cap{L}({S}^{(1)}, \mathbb F)$ and remark that 
  $${\cal E }( \widetilde{K}^{\mathbb F}) {\cal E }(- \widetilde{V}^{(1)}){\cal E}(\psi \is{S}^{(1)})={\cal{E}}(X),$$ where $X$ is an $\mathbb{F}$-semimartingale, and due to (\ref{Equation4.13}) we have 
  $${{{\cal E}(I_{\Rbrack \tau , \infty \Lbrack} \is X)}\over{{\cal E} (-(1-G_{-})^{-1} I_{\Rbrack \tau , \infty \Lbrack} \is m)}}=\widetilde{Z}^{\mathbb{G}}{\cal E}(I_{\Rbrack \tau , \infty \Lbrack}\psi \is{S})$$ is a positive $\mathbb{G}$-supermartinagle. In virtue of Lemma \ref{F-supermartingale2G-supermartingale}, this latter fact is equivalent to ${\cal{E}}(X)$ being a positive $\mathbb{F}$-supermatingale. This proves that $\widetilde{Z}^{\mathbb{F}}\in{\cal D}({S}^{(1)}, \mathbb F)$, and the proof of theorem is complete.
\end{proof}
%%%%%%%%%%%%%%%%%%%%%%%%%%%%%%%%%%%%%%%%%%%%%%%%%%%%%%%%%%%%%
%%%%%%%%%%%%%%%%%%%%%%%%%%%%%%%%%%%%%%%%%%%%%%%%%%%%%%%%%%%%%%%%%%%%%%%%%%%
%%%%%%%%%%%%%%%%%%%%%%%%%%%%%%%%%%%%%%%%%%%%%%%%%%%%%%%%%%%%%%%%%%%%%%%%%%%
The rest of this subsection proves Theorem \ref{riskfactors}, which relies heavily on the following lemma.  
\begin{lemma}\label{lemma6.5}
Suppose that assumptions of Theorem \ref{riskfactors} hold, and consider its notations. Then the following assertions hold.\\
{\rm {(a)}} The $\mathbb{F}$-compensator of $(1-\widetilde{G}) \is \widetilde{\cal H}(\mathbb F)$ is given by 
\begin{align}
     \left ( (1-\widetilde{G}) \is \widetilde{\cal H}(\mathbb F) \right )^{p, \mathbb F} & = (1-G_{-}) \left (\widetilde{\lambda}^{tr} b - \frac{1}{2} \widetilde{\lambda}^{tr} c \widetilde{\lambda} \right) \is A \nonumber \\ 
     & + (1-G_{-}) \left ( \frac{f^{(m)} \widetilde{\lambda}^{tr}x}{(1-\Delta \widetilde{V}^{\mathbb F})(1+\widetilde{\lambda}^{tr}x)} + ( 1-f^{(m)}) \ln (1+\widetilde{\lambda}^{tr}x) - \widetilde{\lambda}^{tr}h \right) \star \nu.
\end{align}
{\rm{(b)}} We have that 
\begin{equation}
   \langle \widetilde{L}^{\mathbb F}, m \rangle^{\mathbb F} = - \widetilde{\lambda}^{tr} c\beta^{(m)} \is A -  \frac{f^{(m)} \widetilde{\lambda}^{tr}x}{(1-\Delta \widetilde{V}^{\mathbb F})(1+\widetilde{\lambda}^{tr}x)}\star\nu.
\end{equation}
{\rm{(c)}} We always have that $E \left [H_T ^{(0)}(\widetilde{K}^{\mathbb G}, \mathbb G) \right ] = E \left [ \left (H^{(0)}(\widetilde{K}^{\mathbb G}, \mathbb G) \right )^{p, \mathbb F}_T \right]$ and 
\begin{align}\label{Equation4.63}
\left (H ^{(0)}(\widetilde{K}^{\mathbb G}, \mathbb G) \right )^{p, \mathbb F} = & \frac{(1-G_{-})}{2} \widetilde{\varphi}^{tr} c \widetilde{\varphi} \is A + \sum_{0 < s\leq .} (1-G_{s-}) \left ( \widetilde{\Gamma}_s^{(1)} - 1 - \ln (\widetilde{\Gamma}_s^{(1)} \right ) \nonumber \\
&+ (1-G_{-}) \left ( \frac{-\widetilde{\Gamma}^{(1)} \widetilde{\varphi}^{tr} x}{1+ \widetilde{\varphi}^{tr} x} + \ln (1+ \widetilde{\varphi}^{tr} x) \right )(1-f^{(m,1)}) \star \nu. \end{align}
and 
\begin{equation}\label{Equation4.64}
    1- \Delta\widetilde{V}^{(1)} = (\widetilde{\Gamma}^{(1)})^{-1}, \quad  and \quad (1-\widetilde{\Gamma}^{(1)}) \frac{(\widetilde{\varphi}^{tr}x)(1-f^{(m,1)})}{1+\widetilde{\varphi}^{tr}x} \star \nu_1 = - \sum \frac{(\widetilde{\Gamma}^{(1)}-1)^2}{\widetilde{\Gamma}^{(1)}}
\end{equation}
\end{lemma}
%%%%%%%%%%%%%%%%%%%%%%%%%%%%%%%%%%%%%%%%%%%%%%%%%%%%%%%%%%%%%%%%%
For the sake of simple exposition, we relegate the proof of this lemma to Appendix \ref{proof4lemmas}.
%%%%%%%%%%%%%%%%%%%%%%%%%%%%%%%%%%%%%%%%%%%%%%%%%%%%%%%%%%%%%%%%%%%%%%%%%%
%%%%%%%%%%%%%%%%%%%%%%%%%%%%%%%%%%%%%%%%%%%%%%%%%%%%%%%%%%%%%%%%%%%%%%%%%
\begin{proof}[Proof of Theorem \ref{riskfactors}]  The proof of the theorem is delivered in two parts, where we prove (\ref{Equation4.17}) and (\ref{Equation4.18}). Both equalities relies on  writing the quantity $u_T(S, \mathbb F)$ in two manners as follows.
\begin{align}
 u_T(S, \mathbb F) &= E \left [ \ln ({\cal E}_T(\widetilde{\lambda} \is S ))\right]  = E \left[ -\ln ({\cal E}_T(-\widetilde{V}^{\mathbb F})) \right ] + E \left[ -\ln ({\cal E}_T(\widetilde{L}^{\mathbb F})) \right ]  \nonumber \\
 & = E \left[\widetilde{V}_T^{\mathbb F} + \sum_{0 < s \leq T} (-\Delta \widetilde{V}_s^{\mathbb F} - \ln (1-\Delta \widetilde{V}_s^{\mathbb F})) + H^{(0)}_T (\widetilde{L}^{\mathbb F}, \mathbb F) \right ]\label{Equation4.65}\\%=E \left [ \widetilde{\cal H}_T (\mathbb F) \right] =
 & =  E \left [ (\widetilde{G} \is \widetilde{\cal H}(\mathbb F))_T \right] + E \left [ ((1-\widetilde{G}) \is \widetilde{\cal H} (\mathbb F))_T \right].\label{Equation4.66}
\end{align}
{\bf Part 1.} This part proves the equality (\ref{Equation4.18}). Thus, we use the duality $\widetilde{Z}^{\mathbb{G}}=1/{\cal{E}}(\widetilde{\varphi} \is (S-S^{\tau}) )$ and (\ref{Equation4.13}), and derive 
\begin{align}
    &u_T(S-S^{\tau}, \mathbb G) \nonumber \\
     & = E \left [ \ln ({\cal E}_T (\widetilde{\varphi} \is (S-S^{\tau}) ))\right]  \nonumber \\
    & = E \left[ -\ln ({\cal E}_T(- I_{\Rbrack \tau , \infty \Lbrack} \is \widetilde{V}^{(a)})) \right ] + E \left[ -\ln ({\cal E}_T(I_{\Rbrack \tau , \infty \Lbrack} \is \widetilde{K}^{\mathbb G})) \right ] \nonumber  \\
    & = E \left[I_{\Rbrack \tau , \infty \Lbrack} \is \widetilde{V}_T^{(1)} + \sum_{0 < s \leq T} I_{\Rbrack \tau , \infty \Lbrack} (-\Delta \widetilde{V}_s^{(1)} - \ln (1-\Delta \widetilde{V}_s^{(1)})) + H^{(0)}_T (\widetilde{K}^{\mathbb G}, \mathbb G) \right ] \label{H(0)4KtildeG} \\
    & = E \left[(1-G_{-}) \is \widetilde{V}_T^{(1)} + (1-G_{-}) \is \sum_{0 < s \leq T} \left (-\Delta \widetilde{V}_s^{(1)} - \ln (1-\Delta \widetilde{V}_s^{(1)}) \right) +\frac{1-G_{-}}{2}\widetilde{\varphi}^{tr} c \widetilde{\varphi} \is A_T  \right] \nonumber \\
    & + E\left[ \sum_{0 < s \leq T} (1-G_{s-}) \left (\widetilde{\Gamma}_s^{(1)} -1 -\ln (\widetilde{\Gamma}_s^{(1)}) \right ) + (1-G_{-})\left ( \ln (1+\widetilde{\varphi}^{tr} x)- \frac{\widetilde{\Gamma}^{(1)} \widetilde{\varphi}^{tr}x}{1+\widetilde{\varphi}^{tr}x}\right )(1-f^{(m,1)})\star \nu_T \right ] \nonumber \\
    & = E \left [ (1-G_{-}) \is \widetilde{V}_T^{(1)} +\frac{1-G_{-}}{2}\widetilde{\varphi}^{tr} c \widetilde{\varphi} \is A_T + (1-G_{-})\left ( \ln (1+\widetilde{\varphi}^{tr} x)- \frac{\widetilde{\varphi}^{tr}x}{1+\widetilde{\varphi}^{tr}x}\right )(1-f^{(m,1)})\star \nu_T \right] \nonumber  \\
    & = E \left [(1-G_{-}) \is \left\{\widetilde{\varphi}^{tr} (b-c(\beta^{(m,1)} + \frac{\widetilde{\varphi}}{2})) \is A_T +\left((1-f^{(m,1)})\ln(1+\widetilde{\varphi}^{tr}x) -\widetilde{\varphi}^{tr} h\right) \star \nu \right\}_T \right].\label{Equation4.67}
\end{align}
Thus, by combining (\ref{Equation4.66}), (\ref{Equation4.67}), Lemma \ref{lemma6.5}-(a) and (\ref{Equation4.19}), we obtain 
\begin{align*}
&u_T(S-S^{\tau}, \mathbb G)-u_T(S, \mathbb F)\\
 &=-E \left [ ((1-\widetilde{G}) \is \widetilde{\cal H} (\mathbb F))_T +\int_0^T{\cal{P}}^{(N,1)}_s dA^{(1)}_s- (1-G_{-})\widetilde{\lambda}^{tr}\beta^{(m,1)}\is A_T- \frac{(1-G_{-})f^{(m,1)} \widetilde{\lambda}^{tr}x}{(1-\Delta \widetilde{V}^{\mathbb F})(1+\widetilde{\lambda}^{tr}x)} \star\nu_T\right].
\end{align*}
Thus, the equality (\ref{Equation4.18}) follows immediately from the above equality combined with Lemma \ref{lemma6.5}-(b) and the facts that $(1-G_{-})\beta^{(m,1)} =\beta^{(m)}I_{\{G_{-}<1\}}$ and $(1-G_{-})f^{(m,1)} =f^{(m)}I_{\{G_{-}<1\}}$.\\
{\bf Part 2.} This parts proves the equality (\ref{Equation4.17}).  To this end, we remark that (\ref{Equation4.13}) implies that ${\cal{E}}( \widetilde{K}^{\mathbb{G}})={\cal{E}}( I_{\Rbrack \tau , \infty \Lbrack}\is\widetilde{K}^{\mathbb{F}})/{\cal{E}}( I_{\Rbrack \tau , \infty \Lbrack}\is{m}^{(1)})$. Hence, by taking logarithm in both sides of the equality and using \cite[Proposition B.2-(a) ]{ChoulliYansori2}, we deduce that 
\begin{equation*}\label{KG2KtildeF}
 H^{(0)} (\widetilde{K}^{\mathbb G}, \mathbb G)=\widetilde{K}^{\mathbb G}- I_{\Rbrack \tau , \infty \Lbrack}\is\widetilde{K}^{\mathbb F}+ I_{\Rbrack \tau , \infty \Lbrack}\is{m}^{(1)}+ I_{\Rbrack \tau , \infty \Lbrack}\is{H}^{(0)} (\widetilde{K}^{\mathbb F}, \mathbb F)-I_{\Rbrack \tau , \infty \Lbrack}\is{H}^{(0)} (m^{(1)}, \mathbb F).
\end{equation*}
Thus, by combining this equality with (\ref{H(0)4KtildeG}), (\ref{Equation4.21}), (\ref{Glocalmartingaleaftertau}) (see Theorem \ref{OptionalDecompoTheorem}-(b)) and (\ref{m(1)2Hellinger}) (see Lemma \ref{deflator4hellinger}-(a)), we derive 
\begin{align*}
     &u_T(S-S^{\tau}, \mathbb G) \\
     & = E \left [((1-\widetilde{G}) \is \widetilde{\cal H} (\mathbb G))_T - I_{\Rbrack \tau , \infty \Lbrack} \is \widetilde{K}_T^{\mathbb F} - \frac{- I_{\Rbrack \tau , \infty \Lbrack}}{1-G_{-}} \is m_T -  (I_{\Rbrack \tau , \infty \Lbrack} \is{H}^{(0)} ({m}^{(1)}, \mathbb{F}))_T \right] \\
     & = E \left [((1-\widetilde{G}) \is \widetilde{\cal H} (\mathbb G))_T + \frac{I_{\Rbrack \tau , \infty \Lbrack}}{1-G_{-}} \is \langle \widetilde{K}^{\mathbb F},m \rangle_T^{\mathbb F}+ \frac{- I_{\Rbrack \tau , \infty \Lbrack}}{(1-G_{-})^2} \is \langle m \rangle^{\mathbb F}_T - H^{(0)}_T (\frac{-I_{\Rbrack \tau , \infty \Lbrack}}{1-G_{-}} \is m, \mathbb F) \right ] \\
     & = E \left [((1-\widetilde{G}) \is \widetilde{\cal H} (\mathbb G))_T +  \langle \widetilde{K}^{\mathbb F}, I_{\{G_{-} <1 \}} \is m \rangle_T^{\mathbb F} \right ] + E \left [ (1-G_{-}) \is  h^{(E)}_T (m^{(1)}, \mathbb F) \right ].
\end{align*}
Therefore, thanks to this latter equality and (\ref{Equation4.66}), we obtain 
\begin{align*}
     &u_T(S-S^{\tau}, \mathbb G)-u_T(S, \mathbb{F})\\
     &=-E \left [ ((1-\widetilde{G}) \is \widetilde{\cal H} (\mathbb F))_T -((1-\widetilde{G}) \is \widetilde{\cal H} (\mathbb G))_T +  \langle \widetilde{L}^{\mathbb F}-\widetilde{K}^{\mathbb F}, I_{\{G_{-} <1 \}} \is m \rangle_T^{\mathbb F} \right ] \\
     &-E \left [ (\widetilde{G} \is \widetilde{\cal H}(\mathbb F))_T \right] + E \left [ \langle \widetilde{L}^{\mathbb F}, I_{\{G_{-} <1 \}} \is m \rangle_T^{\mathbb F}  \right ]+ E \left [ (1-G_{-}) \is  h^{(E)}_T (m^{(1)}, \mathbb F) \right ].
     \end{align*}
     This proves (\ref{Equation4.18}) and the proof of the theorem is complete.
\end{proof}

%%%%%%%%%%%%%%%%%%%%%%%%%%%%%%%%%%%%%%%%%%%%%%%%%%%%%%%%%%%%%%
%%%%%%%%%%%%%%%%%%%%%%%%%%%%%%%%%%%%%%%%%%%%%%%%%%%%%%%%%%%%%%%
%%%\section{Log-optimal portfolio for $(S,\mathbb{G})$}
%%%%%%%%%%%%%%%%%%%%%%%%%%%%%%%%%%%%%%%%%%%%%%%%%%%%%%%%%%%%%%%%%%%%
%%%%%%%%%%%%%%%%%%%%%%%%%%%%%%%%%%%%%%%%%%%%%%%%%%%%%%%%%%%%%%%%%%%%
%%%\begin{theorem}\label{Existence}{\bf EXISTENCE Results}
%%%\end{theorem}
%%%%%%%%%%%%%%%%%%%%%%%%%%%%%%%%%%%%%%%%%%%%%%%%%%%%%%%%%%%%%%%%%%%%
%%%%%%%%%%%%%%%%%%%%%%%%%%%%%%%%%%%%%%%%%%%%%%%%%%%%%%%%%%%%%%%%%%%%
%%%\begin{theorem}\label{Existence}{\bf Description and sensitivity Results}
%%%%\end{theorem}
%%%%%%%%%%%%%%%%%%%%%%%%%%%%%%%%%%%%%%%%%%%%%%%%%%%%%%%%%%%%%%%%%%%%
%%%%%%%%%%%%%%%%%%%%%%%%%%%%%%%%%%%%%%%%%%%%%%%%%%%%%%%%%%%%%%%%%%%%

\appendix
\section{Some results}

 \begin{lemma}\label{G2Fcompensator}
  Suppose that $\tau$ is an honest time. Then for any $\mathbb F$-adapted process $V$ with locally integrable variation, one has 
  \begin{equation}
  I_{\Rbrack \tau, +\infty \Lbrack} \is V^{p, \mathbb G} = 
  I_{\Rbrack \tau, +\infty \Lbrack} (1-G_{-})^{-1} \is \left( (1-\widetilde{G}) \is V \right)^{p, \mathbb F} .
  \end{equation}
 \end{lemma}
 The proof of this lemma can be found in \cite[Lemma 3.1]{aksamitetal18}. Below, we recall the following lemma from \cite[Proposition 5.3]{jeulin80}.
 \begin{lemma}\label{GtoFpredictable}
 Suppose that $\tau$ is an honest time and let $H$ be a $\widetilde{\cal P}(\mathbb G)$-measurable functional. Then, there exists two $\widetilde{\cal P}(\mathbb F)$-measurable functional $J$ and $K$ such that 
 \begin{equation}
     H I_{\Rbrack 0, \infty \Lbrack} = J I_{\Rbrack 0, \tau \Rbrack} + K I_{\Rbrack \tau, \infty \Lbrack}.
 \end{equation}
 \end{lemma}
 
 \begin{lemma}\label{G2exponential} If (\ref{Assumptions4Tau}) holds, then 
 \begin{equation}\label{GTM}
    \dfrac{1-G}{1-G^{\tau}}= {\cal E}(-\frac{1}{1-G_{-}}I_{\Rbrack \tau, \infty \Lbrack} \is m)= {\cal E}(I_{\Rbrack \tau, \infty \Lbrack} \is{m}^{(1)}).
 \end{equation}
\end{lemma}
%%%%%%%%%%%%%%%%%%%%%%%%%%%%%%%%%%%%%%%%%%%%%%%%%%%%%%%%%%%%%%%%%%%%%%%%
The proof of this lemma can be found in \cite[Lemma 4.8]{ChoulliAlharbi}. We end this section by a lemma, which is very useful in some proofs of the main results.
%%%%%%%%%%%%%%%%%%%%%%%%%%%%%%%%%%%%%%%%%%
\begin{lemma}\label{Theta(G)2Theta(F)}  The following assertions hold.\\
{\rm{(a)}} If $V^{\mathbb{G}}$ is RCLL, nondecreasing and $\mathbb{G}$-predictable process such that $(V^{\mathbb{G}})^{\tau}\equiv 0$, then there exists a unique RCLL, nondecreasing and $\mathbb{F}$-predictable process $V^{\mathbb{F}}$ such that 
 \begin{equation}\label{VG2VF}
 I_{\{G_{-}=1\}}\is V^{\mathbb{F}}\equiv 0\quad\mbox{ and}\quad V^{\mathbb{G}}=I_{\Rbrack\tau,+\infty\Lbrack}\is{V}^{\mathbb{F}}.\end{equation}
Furthermore $\Delta{V}^{\mathbb{G}}<1$ if and only $\Delta{V}^{\mathbb{F}}<1$.\\
 {\rm{(b)}} If  (\ref{Assumptions4Tau})  holds, then
 \begin{equation}\label{Equality4SetsTheta}
  {\cal{L}}_b(S-S^{\tau},\mathbb{G})= \Bigl\{\varphi{I}_{\Rbrack\tau,+\infty\Lbrack}\quad:\quad \varphi\in    {\cal{L}}_b({S}^{(1)},\mathbb{F})\Bigr\}.
\end{equation}
 {\rm{(c)}}Let $H$ be any $\mathbb{F}$-predictable process. Then $H\leq 0$ $P\otimes A$-a.e. on $\Rbrack\tau,+\infty\Lbrack$ if and only if  $H\leq 0$ $P\otimes A^{(1)}$-a.e..
 \end{lemma}
 \begin{proof} Both assertion (a) and (b) can be found in  \cite[Lemma 3.9]{ChoulliAlharbi}. Thus the rest of this proof focuses on assertion (c). To this end, we remark that $H\leq 0$ $P\otimes A$-a.e. on $\Rbrack\tau,+\infty\Lbrack$ if and only if 
 $$E\left[\int_0^{\infty} I_{\Rbrack\tau,+\infty\Lbrack}(s) I_{\{H_s>0\}}dA_s\right]=0.$$
 Then, as $A$ and $H$ are $\mathbb{F}$-predictable processes, we take the $\mathbb{F}$-predictable projection inside the RHS term, and conclude that the above equality is equivalent to 
 $$E\left[\int_0^{\infty} (1-G_{s-}) I_{\{H_s>0\}}dA_s\right]=0.$$
 This means that $(1-G_{-}) I_{\{H>0\}}\equiv 0$ $P\otimes A$-a.e., or equivalently $I_{\{H>0\}}\equiv 0$ $P\otimes A^{(1)}$-a.e.. This proves assertion (c), and the proof of the lemma is complete. \end{proof}
%%%%%%%%%%%%%%%%%%%%%%%%%%%%%%%%%%%%%%%%%%%%%%%%%%%%%%%%%%%%%%%%%%%%%%%%%
%%%%%%%%%%%%%%%%%%%%%%%%%%%%%%%%%%%%%%%%%%%%%%%%%%%%%%%%%%%%%%%%%%%%%%
\section{Proofs of Lemmas \ref{deflator4hellinger}, \ref{PredictableCharateristoics4S(tau)}, \ref{F-supermartingale2G-supermartingale} and \ref{lemma6.5}}\label{proof4lemmas}
%%%%%%%%%%%%%%%%%%%%%%%%%%%%%%%%%%%%%%%%%%%%%%%%%%%%%%%%%%%%%%%%%%%%%%%%
%%%%%%%%%%%%%%%%%%%%%%%%%%%%%%%%%%%%%%%%%%%%%%%%%%%%%%%%%%%%%%%%%%%%%%%%%%%%%%%%
%%%%%%%%%%%%%%%%%%%%%%%%%%%%%%%%%%%%%%%%%%%%%%%%%%%%%%%%%%%%%%%%%%%%%%
\begin{proof}[Proof of Lemma \ref{deflator4hellinger}] The proof of the lemma is divided into three parts. The first part proves assertion (a). The second part addresses assertion (b) and (c), while the third part addresses assertion (d).\\
\textbf{Part 1}: This part proves assertion (a). Recall that due to \cite{jeulin80}, $m$ is a BMO $\mathbb{F}$-martingale and hence it is locally bounded, and thanks to \cite[lemma 2.6-(b)]{aksamitetal18} the process $(1-G_{-})^{-1} I_{\{G_{-} < 1\}}$ is locally bounded. Hence, both processes $(1-G_{-})^{-2} I_{\{G_{-}<1\}} \is \langle m \rangle^{\mathbb F}= \langle {m}^{(1)} \rangle^{\mathbb F}$ and $H^{(0)} (m^{(1)},\mathbb{F})$ are nondecreasing  and locally bounded. As a result, the first statement of assertion (a) follows immediately. Furthermore, using the facts that $(A^{p, \mathbb F} = (A^{o, \mathbb F})^{p, \mathbb F})$ for a process $A$ of finite variation,  $(H \is V)^{p,\mathbb F} = {^{p,\mathbb F}(H) \is V}$ for $V$ a predictable of finite variation, $^{p,\mathbb F}(I_{\Rbrack \tau, \infty \Lbrack}) = 1-G_{-}$, and $^{o,\mathbb F}(I_{\Rbrack \tau, \infty \Lbrack}) = 1-\widetilde{G}$,  we derive
\begin{eqnarray*}
&&\left (\frac{I_{\Rbrack \tau,\infty \Lbrack}}{(1-G_{-})^2} \is \langle m \rangle^{\mathbb F} -I_{\Rbrack \tau,\infty \Lbrack} \is H^{(0)} (m^{(1)}, {\mathbb F}) \right )^{p, \mathbb F} \\
&& = \left (\frac{I_{\Rbrack \tau,\infty \Lbrack}}{(1-G_{-})^2} \is \langle m \rangle^{\mathbb F} \right )^{p, \mathbb F} - \left (I_{\Rbrack \tau,\infty \Lbrack} \is H^{(0)} (m^{(1)}, {\mathbb F}) \right )^{p, \mathbb F} \\
&& = \frac{I_{\{G_{-} < 1\}}}{1-G_{-}}\is \langle m \rangle^{\mathbb F} - \left ((1- \widetilde{G}) \is H^{(0)}\left(m^{(1)}, {\mathbb F}\right ) \right )^{p, \mathbb F} \\
&& = \frac{I_{\{G_{-} < 1\}}}{2(1-G_{-})} \is \langle m^{c} \rangle^{\mathbb F} +\sum \frac{(\Delta m)^2}{1-G_{-}}I_{\{G_{-}<1\}} - \Bigl( \sum (1-\widetilde{G}) \left (\frac{-\Delta{m}I_{\{G_{-}<1\}} }{1-G_{-}}- \ln (1 - \frac{\Delta{m}I_{\{G_{-}<1\}} }{1-G_{-}})\right )\Bigr)^{p, \mathbb F} \\ 
&& = \frac{I_{\{G_{-} < 1\}}}{2(1-G_{-})} \is \langle m^{c} \rangle^{\mathbb F} +  (1-G_{-}) \is\left ( \sum\left ( (1+{m}^{(1)}) \ln (1+\Delta{m}^{(1)}) - {m}^{(1)}\right ) \right )^{p, \mathbb F} \\ 
%&& =(1-G_{-}) \is h^{(E)}(- (1-G_{-})^{-1} I_{\{G_{-}<1\}} \is m , \mathbb F) \\
&& =(1-G_{-}) \is h^{(E)}(m^{(1)} , \mathbb F).
\end{eqnarray*}
The last two equalities follow from combining the definition \ref{hellinger} with $\langle m \rangle ^{\mathbb F} = \left ( [m,m]\right )^{p, \mathbb F}$ and $\widetilde{G} =  G_{-} + \Delta m$. This ends the proof of assertion (a).\\
%%%%%%%%%%%%%%%%%%%%%%%%%%%%%%%%%%%%%%%%%%%%%%%%%%%%%%%%%%%%%%%%%%%%%%%%%%%%%%%%%%%
\textbf{Part 2:}  Here we prove assertions (b) and (c). To this end, we consider  $Z^{\mathbb G} \in {\cal D}(S-S^{\tau}, \mathbb G)$ and hence it is clear that $ Z^{\mathbb G}/( Z^{\mathbb G})^{\tau}\in {\cal D}(S-S^{\tau}, \mathbb G)$. Then in virtue of \cite[Theorem 4.4]{ChoulliAlharbi}, we obtain the existence of $Z^{\mathbb F} \in {\cal D}({S}^{(1)}, \mathbb F)$ satisfying
\begin{align*}
    {{Z^{\mathbb G}}\over{( Z^{\mathbb G})^{\tau}}}= \dfrac{Z^{\mathbb F}/(Z^{\mathbb F})^\tau}{ {\cal E}(-(1-G_{-})^{-1} I_{\Rbrack \tau, \infty \Lbrack} \is m)}.
\end{align*}
Then by combining this equality with the fact that $( Z^{\mathbb G})^{\tau}$ is a positive $\mathbb{G}$-supermartingale with $Z^{\mathbb G}_0=1$, we derive
\begin{eqnarray*}
E[ - \ln Z^{\mathbb G}_T] &&= E[ - \ln Z^{\mathbb G}_T - \ln Z^{\mathbb G}_{T\wedge\tau}]+E[ - \ln Z^{\mathbb G}_{T\wedge\tau}]\\
&&\geq  E\Bigl[ - \ln \left({{Z^{\mathbb G}_T}\over{Z^{\mathbb G}_{T\wedge\tau}}}\right)\Bigr]=E\Bigl[-\ln \dfrac{Z^{\mathbb F}_T/(Z^{\mathbb F}_{T\wedge\tau}}{{\cal E}_T(-(1-G_{-})^{-1} I_{\Rbrack \tau, \infty \Lbrack} \is m)}\Bigr].\end{eqnarray*}
This proves assertion (b). To prove assertion (c), on the one hand, we remark that assertion (b) implies clearly that
\begin{eqnarray*} \inf_{Z^{\mathbb G} \in {\cal D}(S-S^{\tau}, \mathbb G)} E \left [ - \ln Z^{\mathbb G} \right ] \geq \inf_{Z^{\mathbb F} \in {\cal D}({S}^{(1)}, \mathbb F)} E \left [-\ln \dfrac{Z^{\mathbb F}/(Z^{\mathbb F})^\tau}{ {\cal E}(I_{\Rbrack \tau, \infty \Lbrack} \is{m}^{(1)})}\right ].
\end{eqnarray*}
On the other hand, again \cite[Theorem 4.4]{ChoulliAlharbi} allows us to conclude that 
$$\Biggl\{\dfrac{Z^{\mathbb F}/(Z^{\mathbb F})^\tau}{ {\cal E}(I_{\Rbrack \tau, \infty \Lbrack} \is{m}^{(1)})}:\ Z^{\mathbb F}\in  {\cal D}({S}^{(1)}, \mathbb F)\Biggr\}\subset  {\cal D}({S}-S^{\tau}, \mathbb{G}),$$ and hence 
\begin{eqnarray*} \inf_{Z^{\mathbb G} \in {\cal D}(S-S^{\tau}, \mathbb G)} E \left [ - \ln Z^{\mathbb G} \right ] \leq \inf_{Z^{\mathbb F} \in {\cal D}({S}^{(1)}, \mathbb F)} E \left [-\ln \dfrac{Z^{\mathbb F}/(Z^{\mathbb F})^\tau}{ {\cal E}(I_{\Rbrack \tau, \infty \Lbrack} \is{m}^{(1)})}\right ].
\end{eqnarray*}
Therefore, assertion (c) follows from the last two inequalities. \\
\textbf{Part 3:}  To prove assertion (d), we apply \cite[Proposition B.2-(a)]{ChoulliYansori2} repeatedly whenever it is convenient, and get   derive 
\begin{eqnarray}
&&-\ln \left (\frac{Z^{\mathbb F}/(Z^{\mathbb F})^{ \tau }}{{\cal E} \left (- I_{\Rbrack \tau, \infty \Lbrack}(1-G_{-})^{-1}  \is m \right)}\right ) \nonumber\\
&& = -\ln {\cal E}(I_{\Rbrack \tau, \infty \Lbrack} \is K^{\mathbb F}) +\ln {\cal E}(-I_{\Rbrack \tau, \infty \Lbrack}(1-G_{-})^{-1} \is m) + {I_{\Rbrack \tau, \infty \Lbrack} \is V}\nonumber \\ 
&& = -I_{\Rbrack \tau, \infty \Lbrack} \is K^{\mathbb F} + I_{\Rbrack \tau, \infty \Lbrack} \is H^{(0)} (K^{\mathbb F}, \mathbb F) + I_{\Rbrack \tau, \infty \Lbrack} \is V - I_{\Rbrack \tau, \infty \Lbrack} (1-G_{-})^{-1} \is m - I_{\Rbrack \tau, \infty \Lbrack} \is H^{(0)}(m^{(1)}, \mathbb F) \nonumber\\ 
&& = {\mathbb G}\mbox{-local martingale} + I_{\Rbrack \tau, \infty \Lbrack} \is V+ \frac{I_{\Rbrack \tau, \infty \Lbrack} }{1-G_{-}} \is \langle K^{\mathbb F},  m \rangle^{\mathbb F} + I_{\Rbrack \tau, \infty \Lbrack} \is H^{(0)}(K^{\mathbb F}, \mathbb F)\nonumber \\
&& + \frac{ I_{\Rbrack \tau, \infty \Lbrack}}{(1-G_{-})^2} \is \langle m, m \rangle^{\mathbb F} - I_{\Rbrack \tau, \infty \Lbrack}\is H^{(0)}(m^{(1)} , \mathbb F)\nonumber\\
&& = {\mathbb G}-\mbox{local martingale} + I_{\Rbrack \tau, \infty \Lbrack} \is V +I_{\Rbrack \tau, \infty \Lbrack} \is \langle K^{\mathbb F}, m^{(1)} \rangle^{\mathbb F} + \frac{I_{\Rbrack \tau, \infty \Lbrack}}{(1-G_{-})^2}  \is \langle m \rangle^{\mathbb F}\nonumber \\
&& + \left ( I_{\Rbrack \tau, \infty \Lbrack} \is \left ( H^{(0)}(K^{\mathbb F}, \mathbb F) - H^{(0)}(m^{(1)}, \mathbb F)\right ) \right )^{p, \mathbb G}\nonumber\\
&&= {\mathbb G}\mbox{-local martingale} + I_{\Rbrack \tau, \infty \Lbrack} \is \left(V +\langle K^{\mathbb F}, m^{(1)} \rangle^{\mathbb F} +{h}^{(E)}(m^{(1)}, \mathbb F)+ \left ({H}^{(0)}(K^{\mathbb F}, \mathbb F)\right )^{p, \mathbb G}\right).\label{Equa100}
\end{eqnarray}
The third and fourth equalities are due to \eqref{Glocalmartingaleaftertau} applied to both $K^{\mathbb F}$ and $m$, and to the fact that $V- V^{p,\mathbb G}$ is $\mathbb G$-local martingale for any $V \in {\cal A}_{loc}(\mathbb G)$ respectively. The last equality follows from  assertion (a). Thus the decomposition (\ref{Equa100}) is the Doob-Meyer's decomposition of the positive $\mathbb{G}$-submartingale $-\ln \left (Z^{\mathbb F}/(Z^{\mathbb F})^{ \tau }{\cal E} \left (- I_{\Rbrack \tau, \infty \Lbrack}(1-G_{-})^{-1}  \is m \right)\right ) $. Then by applying \cite[Proposition B.2-(b)]{ChoulliYansori2} to it, the proof of assertion (d) follows immediately, and the proof of the lemma is complete.
\end{proof}
%%%%%%%%%%%%%%%%%%%%%%%%%%%%%%%%%%%%%%%%%%%%%%%%%%%%%%%%%%%%%%%%%%%%%%%%
%%%%%%%%%%%%%%%%%%%%%%%%%%%%%%%%%%%%%%%%%%%%%%%%%%%%%%%%%%%%%%%%%%%%%%%%%%
%%%%%%%%%%%%%%%%%%%%%%%%%%%%%%%%%%%%%%%%%%%%%%%%%%%%%%%%
 \begin{proof}[Proof of Lemma \ref{F-supermartingale2G-supermartingale}] As $X$ is an $\mathbb{F}$-semimartingale, then there exist $M\in {\cal{M}}_{loc}(\mathbb{F})$ and $B$ a RCLL $\mathbb{F}$-adapted process with finite variation such that 
 $$X=X_0+M-B,\quad\mbox{and}\quad M_0=B_0=0.$$
 Remark that, due to the fact that $B$ has a finite variation, we have 
 $${\cal T}^{(a)}(B) = (1-G_{-})(1-\widetilde{G} )^{-1}I_{\Rbrack \tau, \infty \Lbrack} \is B.$$
 Then by combining this with \cite[Lemma 4.3-(a)]{ChoulliAlharbi}, we obtain 
  \begin{align}\label{Y2MandB}
     Y &:= {{{\cal E}(I_{\Rbrack \tau , \infty \Lbrack} \is X)}\over{{\cal E} (-(1-G_{-})^{-1} I_{\Rbrack \tau , \infty \Lbrack} \is m )} }=  {\cal E}\left ({\cal T}^{(a)} (X)+(1-G_{-})^{-1}I_{\Rbrack \tau , \infty \Lbrack}  \is {\cal T}^{(a)} (m) \right)\nonumber\\
   &=   {\cal E}\left ({\cal T}^{(a)} (M)+(1-G_{-})^{-1}I_{\Rbrack \tau , \infty \Lbrack}  \is {\cal T}^{(a)} (m) -(1-G_{-})(1-\widetilde{G})^{-1}I_{\Rbrack \tau , \infty \Lbrack} \is{B}\right).
 \end{align}
1) If ${\cal E}(X)$ is a positive $\mathbb F$-supermartingale, or equivalently $X$ is an $\mathbb F$- local supermartingale,  then $B$ has an $\mathbb{F}$-locally integrable variation. Furthermore, $B^{p,\mathbb{F}}$ is nondecreasing and 
$$ {\cal T}^{(a)} (B)= {\cal T}^{(a)} (B-B^{p,\mathbb{F}})+{\cal T}^{(a)} (B^{p,\mathbb{F}})={\cal T}^{(a)} (B-B^{p,\mathbb{F}})+{{1-G_{-}}\over{1-\widetilde{G}}}I_{\Rbrack \tau , \infty \Lbrack} \is {B}^{p,\mathbb{F}}.$$
Thus, by inserting this in (\ref{Y2MandB}), we get 
$$
Y= {\cal E}\left ({\cal T}^{(a)} (M+B-B^{p,\mathbb{F}})+{{I_{\Rbrack \tau , \infty \Lbrack} }\over{1-G_{-}}} \is {\cal T}^{(a)} (m) -{{1-G_{-}}\over{1-\widetilde{G}}}I_{\Rbrack \tau , \infty \Lbrack} \is{B}^{p,\mathbb{F}}\right).$$
This proves that $Y$ is a nonnegative $\mathbb{G}$-local supermartingale, and hence it is a $\mathbb{G}$-supermartingale.\\
To prove the reverse, we assume that $Y$ is a a $\mathbb{G}$-supermartingale, or equiavelntly 
$$K:={\cal T}^{(a)} (M)+(1-G_{-})^{-1}I_{\Rbrack \tau , \infty \Lbrack}  \is {\cal T}^{(a)} (m) -(1-G_{-})(1-\widetilde{G})^{-1}I_{\Rbrack \tau , \infty \Lbrack} \is{B},$$ is a $\mathbb{G}$-local supermartingale. This is equivalent to $ (1-G_{-})(1-\widetilde{G})^{-1}I_{\Rbrack \tau , \infty \Lbrack} \is{B}$ being a $\mathbb{G}$-local submartingale with finite variation. Hence it has a $\mathbb{G}$-local integarble variation and its compensator is nondecreasing, i.e.
$$\left((1-G_{-})(1-\widetilde{G})^{-1}I_{\Rbrack \tau , \infty \Lbrack} \is{B}\right)^{p,\mathbb{G}}=I_{\Rbrack \tau , \infty \Lbrack} \is\left(I_{\{\widetilde{G}<1\}}\is{B}\right)^{p,\mathbb{F}}=I_{\Rbrack \tau , \infty \Lbrack} \is{B}^{p,\mathbb{F}},$$ is nondecreasing. This implies $B\in {\cal A}_{loc}(\mathbb{F})$ and $B^{p,\mathbb{F}}$ is nondecreasing, and hence $X$ is an $\mathbb{F}$-local supermartingale. This completes the proof of the lemma. 
 \end{proof}
 %%%%%%%%%%%%%%%%%%%%%%%%%%%%%%%%%%%%%%%%%%%%%%%%%%%%%%%%%%%%%%%%%%%%%%%%%%%%%%
%%%%%%%%%%%%%%%%%%%%%%%%%%%%%%%%%%%%%%%%%%%%%%%%%%%%%%%%%%%%%%%%%%%%%%%%%%%%
\begin{proof}[Proof of Lemma \ref{lemma6.5}] The proof of the lemma is divided into three parts, where we prove the three assertions respectively. \\
{\bf Part 1.} Here we prove assertion (a). To this end, we remark the following
$$
(1-\widetilde{G})I_{\{\Delta{S}\not=0\}}=(1-G_{-})\left(1-f^{(m,1)}(\Delta{S})-g^{(m,1)}(\Delta{S})\right)I_{\{\Delta{S}\not=0\}},$$
$$ (1-\widetilde{G})I_{\{\Delta{S}=0\}}=(1-G_{-})\left(1-{{\widehat{f^{(m,1)}}}\over{1-a}}-m^{(\perp,1)}\right)I_{\{\Delta{S}=0\}},$$
 and 
 $$
 \Delta\widetilde{L}^{\mathbb{F}}={{\widetilde{\Xi}}\over{1+\widetilde{\lambda}\Delta{S}}}-1.$$
 Thus, by combining these remarks, we calculate the following 
\begin{align*}
 & (1- \widetilde{G}) \is H^{(0)} (\widetilde{L}^{\mathbb F}, \mathbb F)\\
   & = (1-\widetilde{G}) \is\langle (\widetilde{L}^{\mathbb F})^c \rangle+ \sum (1-\widetilde{G})\left (\Delta \widetilde{L}^{\mathbb F} - \ln (1+\Delta \widetilde{L}^{\mathbb F}) \right) \\
  & = (1-\widetilde{G}) \is\langle (\widetilde{L}^{\mathbb F})^c \rangle+ \sum (1-\widetilde{G})\left (\Delta \widetilde{L}^{\mathbb F} - \ln (1+\Delta \widetilde{L}^{\mathbb F}) \right)I_{\{\Delta{S}\not=0\}}\\
  &+ \sum (1-\widetilde{G})\left (\Delta \widetilde{L}^{\mathbb F} - \ln (1+\Delta \widetilde{L}^{\mathbb F}) \right)I_{\{\Delta{S}=0\}}\\
  & = \frac{(1-{G}_{-})}{2} \widetilde{\lambda}^{tr} c \widetilde{\lambda} \is A + (1-G_{-})\left(1-f^{(m,1)}-g^{(m,1)}\right)\left({{\widetilde{\Xi}}\over{1+\widetilde{\lambda}^{tr}x}}-1-\ln({{\widetilde{\Xi}}\over{1+\widetilde{\lambda}^{tr}x}})\right)\star\mu\\
  &+ \sum (1-G_{-})\left(1+{{\widehat{f^{(m,1)}}}\over{1-a}}-m^{(\perp,1)}\right)(\widetilde{\Xi}-1-\ln(\widetilde{\Xi}))I_{\{\Delta{S}=0\}}\\
  & = \frac{(1-{G}_{-})}{2} \widetilde{\lambda}^{tr} c \widetilde{\lambda} \is A + (1-G_{-})\left(1-f^{(m,1)}\right)\left(-{{\widetilde{\Xi}(\widetilde{\lambda}^{tr}x)}\over{1+\widetilde{\lambda}^{tr}x}}+\ln(1+\widetilde{\lambda}^{tr}x)\right)\star\mu\\
  %&+ \sum (1-G_{-})\left(1-{{\widehat{f^{(m,1)}}}\over{1-a}}-m^{(\perp,1)}\right)(\widetilde{\Xi}-1-\ln(\widetilde{\Xi}))I_{\{\Delta{S}=0\}}\\
  &\ +\sum (1-G_{-})(\widetilde{\Xi}-1-\ln(\widetilde{\Xi}))+\sum (1-G_{-}){{\widehat{f^{(m,1)}}}\over{1-a}}(\widetilde{\Xi}-1-\ln(\widetilde{\Xi}))I_{\{\Delta{S}=0\}}\\
  &\ - (1-G_{-})(\widetilde{\Xi}-1-\ln(\widetilde{\Xi}))f^{(m,1)}\star\mu+\mbox{local martingale}
  \end{align*}
  
Then by compensating on both sides and using $\Delta{V}^{\mathbb{F}}=1-(\widetilde{\Xi})^{-1}$, we get 
  \begin{align*}
  &\left ((1- \widetilde{G}) \is H^{(0)} (\widetilde{L}^{\mathbb F}, \mathbb F) \right)^{p, \mathbb F} \\
  & = \frac{(1-G_{-})}{2} \widetilde{\lambda}^{tr} c \widetilde{\lambda} \is A +(1-G_{-})(1-f^{(m,1)})\left (\frac{- \widetilde{\lambda}^{tr}x}{1+\widetilde{\lambda}^{tr}x} +\ln (1+\widetilde{\lambda}^{tr}x)\right ) \star \nu\\
  & + \sum (1-G_{-}) \left (\frac{\Delta \widetilde{V}^{\mathbb F}}{1-\Delta \widetilde{V}^{\mathbb F}} + \ln (1-\Delta \widetilde{V}^{\mathbb F}) \right) + (1-G_{-})\frac {(f-1)(1-f^{(m,1)}) \Delta \widetilde{V}^{\mathbb F}}{1-\Delta \widetilde{V}^{\mathbb F}} \star \nu.
\end{align*}
Then, we can calculate 
\begin{align*}
   \left ((1- \widetilde{G}) \is \widetilde{\cal H}(\mathbb F) \right)^{p, \mathbb F} & = (1-G_{-}) \is \widetilde{V}^{\mathbb F} + \sum (1-G_{-}) \left (-\Delta \widetilde{V}^{\mathbb F} -\ln (1- \Delta \widetilde{V}^{\mathbb F}) \right ) + \left ((1-\widetilde{G}) \is H^{(0)} (\widetilde{L}^{\mathbb F}, \mathbb F)\right)^{p, \mathbb F}\\
   & = (1-G_{-})\is \widetilde{V}^{\mathbb F} + \sum (1-G_{-}) \left (-\Delta \widetilde{V}^{\mathbb F} -\ln (1- \Delta \widetilde{V}^{\mathbb F}) +\frac{\Delta \widetilde{V}^{\mathbb F}}{1-\Delta \widetilde{V}^{\mathbb F}} + \ln (1-\Delta \widetilde{V}^{\mathbb F} )\right )\\
   & + \frac{(1-G_{-})}{2} \widetilde{\lambda}^{tr} c \widetilde{\lambda} \is A + (1-G_{-})(1-f_m^{(a)})\left (\frac{-\widetilde{\lambda}^{tr}x}{1+\widetilde{\lambda}^{tr}x} +\ln (1+\widetilde{\lambda}^{tr}x)\right ) \star \nu \\
   & + (1-G_{-})\frac {(f-1)(1-f^{(1)}) \Delta \widetilde{V}^{\mathbb F}}{1-\Delta \widetilde{V}^{\mathbb F}} \star \nu \\
   & = (1-G_{-}) \left(\widetilde{\lambda}^{tr}b-\frac{1}{2}\widetilde{\lambda}^{tr}c\widetilde{\lambda} \right ) \is A \\
   & + (1-G_{-})\left(\frac{f^{(m,1)}\widetilde{\lambda}^{tr}x}{1+\widetilde{\lambda}x} + (1-f^{(m,1)})\ln(1+\widetilde{\lambda}^{tr}x)-\widetilde{\lambda}^{tr}h \right )\star \nu \\
   & + (1-G_{-}) \left (\frac{(1-f)(a-\hat{f})}{1-\Delta \widetilde{V}^{\mathbb F}} \right)\star \nu + (1-G_{-})(1-f_m^{(a)}) \left(\frac{(a-\hat{f})(f-1)}{1-\Delta \widetilde{V}^{\mathbb F}} \right) \star \nu \\
   & = (1-G_{-}) \left(\widetilde{\lambda}^{tr}b-\frac{1}{2}\widetilde{\lambda}^{tr}c\widetilde{\lambda} \right ) \is A \\
   & + (1-G_{-}) \left( \frac{f^{(m,1)}\widetilde{\lambda}^{tr}x}{(1-\Delta \widetilde{V}^{\mathbb F})(1+\widetilde{\lambda}^{tr}x)} + (1-f^{(m,1)})\ln(1+\widetilde{\lambda}^{tr}x)-\widetilde{\lambda}^{tr} h \right) \star \nu. 
\end{align*}
This proves assertion (a) an completes part 1.\\
{\bf Part 2.} This part proves assertion (b). We remark that due to (\ref{Decomposition4m})  and (\ref{Equation4.20}) we have 
$$\Delta\widetilde{L}^{\mathbb F}=\widetilde{\Xi}(1+\widetilde{\lambda}^{tr}\Delta{S})^{-1}-1\quad\mbox{and}\quad \Delta{m}=\left(f^{(m)}(\Delta{S})+g^{(m)}(\Delta{S})\right)I_{\{\Delta{S}\not=0\}}+{{\widehat{f^{(m)}}}\over{1-a}}I_{\{\Delta{S}=0\}} +\Delta{m}^{\perp}.$$
Thus, by using these equalities, we derive  
\begin{align}
  [\widetilde{L}^{\mathbb F}, m]&= -\widetilde{\lambda}^{tr}c\beta^{(m)} \is A + \sum \Delta{m}\Delta\widetilde{L}^{\mathbb F}I_{\{\Delta{S}\not=0\}}+ \sum \Delta{m}\Delta\widetilde{L}^{\mathbb F}I_{\{\Delta{S}=0\}}\nonumber\\
  & =  -\widetilde{\lambda}^{tr}c\beta^{(m)} \is A + \left (f^{(m)} + g^{(m)}\right)({{\widetilde{\Xi}}\over{1+\widetilde{\lambda}^{tr}x}}-1)\star\mu+ (\widetilde{\Xi}-1)\is{m}^{\perp}\label{Equation(E100)} \\
  &\hskip 1cm- \sum \frac{\widehat{f_m}^{(a)}}{1-a}(\widetilde{\Xi}-1)I_{\{a<1\}}I_{\{\Delta S= 0\}}. \nonumber\end{align}
  Then remark that it is easy to prove that both processes $ f^{(m)}({\widetilde{\Xi}}(1+\widetilde{\lambda}^{tr}x)^{-1}-1)\star(\mu-\nu)$  and $ g^{(m)}({\widetilde{\Xi}}(1+\widetilde{\lambda}^{tr}x)^{-1}-1)\star\mu$ are $\mathbb{F}$-local martingales. Hence, by compensating in both sides of (\ref{Equation(E100)}) and using the latter remark, we get 
  \begin{align*}
    \langle\widetilde{L}^{\mathbb F}, m\rangle^{\mathbb{F}}&= -\widetilde{\lambda}^{tr}c\beta^{(m)} \is A + f^{(m)}({{\widetilde{\Xi}}\over{1+\widetilde{\lambda}^{tr}x}}-1)\star\nu- \underbrace{\sum \widehat{f^{(m)}}(\widetilde{\Xi}-1)}_{=(\widetilde{\Xi}-1)f^{(m)}\star\nu}\\
    &=-\widetilde{\lambda}^{tr}c\beta^{(m)} \is A - f^{(m)}\left({{\widetilde{\Xi}\widetilde{\lambda}^{tr}x}\over{1+\widetilde{\lambda}^{tr}x}}-1\right)\star\nu\end{align*}
  This ends the proof of assertion (b).\\
  %& \left(f(1-a+\hat{f})^{-1} I_{\{\Delta S \neq 0 \}}-1+ (1-a+\hat{f})^{-1} I_{\{\Delta S = 0 \}} \right) \\
%  & = -\widetilde{\lambda}^{tr}c\beta_m^{(a)} \is A + \sum (f_m^{(a)}+ g_m^{(a)})\left (\frac{f}{1-a+\hat{f}}-1 \right ) I_{\{\Delta S \neq 0\}} \\
 % & + \sum (\Delta m^{(\perp,a)} -\frac{\widehat{f_m}^{(a)}}{1-a})\left (\frac{1}{1-a+\hat{f}}-1 \right ) I_{\{\Delta S = 0\}} \\
 % & = -\widetilde{\lambda}^{tr}c\beta_m^{(a)} \is A + \sum \left (\frac{f-1}{1-a+\hat{f}}- \frac{\hat{f}-1}{1-a+\hat{f}}\right)(f_m^{(a)}+g_m^{(a)})I_{\{\Delta S \neq 0\}}\\
%  & + \sum \frac{\widehat{f_m}^{(a)}(\hat{f}-a)}{(1-a+\hat{f})(1-a)} I_{\{\Delta S = 0\}} - \sum \frac{\hat{f}-a}{1-a+\hat{f}} \Delta m^{(\perp,a)}\\
 % & =  -\widetilde{\lambda}^{tr}c\beta_m^{(a)} \is A + \left (\frac{f-1}{1-a+\hat{f}}-\frac{\hat{f}-a}{1-a+\hat{f}} \right )(f_m^{(a)}+g_m^{(a)}) \star \mu \\
  %& + \frac{\hat{f}-a}{1-a+\hat{f}}f_m^{(a)} \star \nu -\frac{\hat{f}-a}{1-a+\hat{f}} \is m^{(\perp, a)}
%%Then by compensating both sides of the above equality, \
{\bf Part 3.} Here we prove assertion (c). Thus, we recall 
$$\Delta \widetilde{K}^{\mathbb G} = I_{\Rbrack \tau, \infty \Lbrack} \left ( \dfrac{\widetilde{\Gamma}^{(a)}}{1+\widetilde{\varphi}^{tr} \Delta S} -1\right),$$  and derive 
\begin{align*}
    &H^{(0)} (\widetilde{K}, \mathbb G) \\
    & = \frac{1}{2} \langle(\widetilde{K}^{\mathbb G})^c \rangle + \sum \left ( \Delta \widetilde{K}^{\mathbb G} -\ln (1+ \Delta \widetilde{K}^{\mathbb G}) \right ) \\
    & = \frac{1}{2}\widetilde{\varphi}^{tr} c \widetilde{\varphi} I_{\Rbrack \tau, \infty \Lbrack} \is A + \sum I_{\Rbrack \tau, \infty \Lbrack} \left (  \left (\dfrac{\widetilde{\Gamma}^{(a)}}{1+\widetilde{\varphi}^{tr} \Delta S} -1\right ) -\ln \left (\frac{\widetilde{\Gamma}^{(1)}}{1+\widetilde{\varphi}^{tr} \Delta S}\right )\right )\\
    & = \frac{1}{2}\widetilde{\varphi}^{tr} c \widetilde{\varphi} I_{\Rbrack \tau, \infty \Lbrack} \is A +\sum I_{\Rbrack \tau, \infty \Lbrack} \left (\widetilde{\Gamma}^{(1)} -1 -\ln (\widetilde{\Gamma}^{(1)}) \right )+ \sum I_{\Rbrack \tau, \infty \Lbrack} \left ( \ln (1+\widetilde{\varphi}^{tr} \Delta S)-  \frac{\widetilde{\Gamma}^{(1)} \widetilde{\varphi}^{tr} \Delta S}{1+\widetilde{\varphi}^{tr} \Delta S}\right )\\
    & = \frac{1}{2}\widetilde{\varphi}^{tr} c \widetilde{\varphi} I_{\Rbrack \tau, \infty \Lbrack} \is A +\sum I_{\Rbrack \tau, \infty \Lbrack} \left (\widetilde{\Gamma}^{(1)} -1 -\ln (\widetilde{\Gamma}^{(1)}) \right )+ \left ( \ln (1+\widetilde{\varphi}^{tr} x)-  \frac{\widetilde{\Gamma}^{(1)} \widetilde{\varphi}^{tr}x}{1+\widetilde{\varphi}^{tr}x}\right )I_{\Rbrack \tau, \infty \Lbrack} \star \mu. 
\end{align*}
Then (\ref{Equation4.63}) follows immediately from taking the $\mathbb F$-dual predictable projection in both sides of the above equality, and using afterwards the two facts $(I_{\Rbrack \tau, \infty \Lbrack} \is U)^{p,\mathbb{F}}=(1-G_{-})\is V$ for any $\mathbb{F}$-predictable and nondecreasing process and $(I_{\Rbrack \tau, \infty \Lbrack}W\star\mu)^{p,\mathbb{F}}=(1-G_{-})(1-f^{(m,1)})W\star\nu$ for any nonnegative and $\widetilde{\cal{P}}(\mathbb{F})$-measurable $W$. The equalities in (\ref{Equation4.64}) follow from direct calculations using (\ref{V(1)}) and (\ref{Gamma(1)andf(op)}) . This proves assertion (c), and the proof of the lemma is complete.
\end{proof}

 %%%%%%%%%%%%%%%%%%%%%%%%%%%%%%%%%%%%%%%%%%%%%%%%%%%%%%%%%%%%%%%%%%%%%%%%%%%%%%%%%
 %%%%%%%%%%%%%%%%%%%%%%%%%%%%%%%%%%%%%%%%%%%%%%%%%%%%%%%%%%%%%

\end{document}